\documentclass[11pt]{article}

\usepackage{mathrsfs}
\usepackage{bbm}
\usepackage{color}
\usepackage[top=1in, bottom=1in, left=1in, right=1in]{geometry}
\usepackage{amsfonts}
\usepackage{verbatim}
\usepackage{dsfont}
\usepackage{mathrsfs}
\usepackage{graphicx,amscd,xypic}
\usepackage{amsthm}
\usepackage{amssymb,amsmath,epsfig,natbib}
\usepackage{enumerate,enumitem}
\usepackage[colorlinks=true,citecolor=blue]{hyperref}
\usepackage{tikz}
\usepackage{ragged2e}

\numberwithin{equation}{section}

\newtheorem{theorem}{Theorem}[section]
\newtheorem{lemma}[theorem]{Lemma}
\newtheorem{proposition}[theorem]{Proposition}
\newtheorem{corollary}[theorem]{Corollary}

\newtheorem{Assumption}{Assumption}[section]

\renewcommand{\raggedright}{\leftskip=12pt \rightskip=12pt plus 0cm}

\newcommand{\EE}{{\mathbb E}}
\newcommand{\II}{{\mathbb I}}
\newcommand{\RR}{{\mathbb R}}
\newcommand{\PP}{{\mathbb P}}

\newcommand{\dd}{{\mathrm d}}

\allowdisplaybreaks

\title{\bf Variance Contracts}
\author{Yichun Chi\thanks{China Institute for Actuarial Science, Central University of Finance and Economics, Beijing 102206, China. Email: \url{yichun@cufe.edu.cn}.} \and Xun Yu Zhou\thanks{Department of Industrial Engineering and Operations Research and The Data Science Institute, Columbia University, New York, NY 10027, USA. Email: \url{xz2574@columbia.edu}.}
	\and Sheng Chao Zhuang\thanks{Department of Finance, University of Nebraska-Lincoln, NE, USA. Email: \url{szhuang3@unl.edu}.}}
\date{}

\begin{document}
\maketitle \vspace*{0.5cm}

\begin{abstract}
\normalsize

We study the design of an optimal insurance contract in which the insured maximizes her expected utility and the insurer limits the variance of his risk exposure while maintaining the principle of indemnity and charging the premium according to the expected value principle.  We derive the optimal policy semi-analytically, which is coinsurance above a deductible when the variance bound is binding. This policy automatically satisfies the incentive-compatible condition, which is crucial to rule out  ex post moral hazard.
We also find that the deductible is absent if and only if the contract pricing is actuarially fair.
Focusing on the actuarially fair case, we carry out comparative statics on the effects of the insured's initial wealth and the variance bound on insurance demand. Our results indicate that the expected coverage
is always larger for a wealthier insured, implying that the underlying  insurance is a {\it normal} good, which supports certain recent empirical findings.
Moreover, as the variance constraint tightens, the insured who is prudent cedes less losses, while the insurer is exposed to less tail risk.

\end{abstract}

\vspace*{0.5cm}

\begin{bfseries}Key-words\end{bfseries}: Insurance design; expected value principle; variance; incentive compatibility; comparative statics.
\newpage

\section{Introduction}

Insurance is an efficient mechanism to facilitate risk reallocation between two parties.
\citet{Borch1960} was the first to study the
insurance contract design problem and to prove
that given a fixed
premium, a stop-loss (or deductible) insurance policy (i.e., full coverage above a deductible) achieves the smallest variance of the insured's share of  payment.\footnote{\citet{Borch1960} presents the problem in a reinsurance setting, in which the ceding insurer corresponds to the insured in an insurance setting.}
\citet{Arrow1963} assumes that the premium is calculated by the {\it expected value principle} (i.e., the insurance cost is proportional to the expected indemnity) and
imposes the \emph{principle of indemnity} (i.e., the insurer's reimbursement  is non-negative and smaller than the loss). Under these specifics, \citet{Arrow1963} shows that the stop-loss insurance is Pareto optimal between a risk-neutral insurer and a risk-averse insured. This is a foundational result that has earned the name of {\it Arrow's theorem of the deductible} in the literature. \citet{Mossin1968} further proves that the deductible is strictly positive if and only if the insurance price is actuarially unfair (i.e., the safety loading is strictly positive).

In \citet{Arrow1963}, the insurer is assumed to be risk-neutral. This is based on the assumption that the insurer has a sufficiently large number of independent and homogenous insureds, such that his risk, by the law of large numbers,  is sufficiently diversified to be nearly zero. This kind of theoretically
ideal situation hardly occurs in practice, even if the insurer does indeed have
a huge number of clients. Moreover, it does not apply to tailor-made contracts for insuring one-off events (e.g., the shipment of a highly valuable painting). Theorem 2 in \citet{Arrow1971} stipulates  that, when the insurer is risk-averse and  insurance cost is absent, an optimal contract must involve coinsurance.   \citet{Raviv} extends this result to include nonlinear insurance costs and shows that an optimal policy involves
both a deductible and coinsurance. Much of the recent research in this area has focused on the insurer's tail risk exposure. \citet{Cummins2004} and \citet{Zhou2010} extend Arrow's model by introducing an {\it exogenous} upper bound on indemnity and thereby limiting the insurer's liability with respect to catastrophic losses. From a regulatory perspective, \citet{ZhouWu2008} propose a model in which the insurer's expected loss
above a prescribed level
is controlled, and they conclude that an optimal policy is generally piecewise linear.
\citet{Doherty2015} investigate the case in which losses are nonverifiable and deduce that a contract with a deductible and an {\it endogenous} upper limit is optimal.

As important as tail risks are for both parties in insurance contracts,  in practice insurers are also concerned with other parts of the loss distribution. \citet{Kaye2005}, in a Casualty Actuarial Society report, writes \vspace{0.4cm}
\begin{center}
\parbox{\textwidth}{\raggedright{\it ``Different stakeholders have different levels of interest in different parts of the distribution - the perspective of the decision-maker is important. Regulators and rating agencies will be focused on the extreme downside where the very existence of the company is in doubt. On the other hand, management and investors will have a greater interest in more near-term scenarios towards the middle of the distribution and will focus on the likelihood of making a profit as well as a loss" (p. 4). }}
\end{center}
\vspace{0.4cm}

Assuming the insurer to be risk-averse with a concave utility function  indeed takes the whole risk distribution into consideration, as studied in  \citet{Arrow1971} and
\citet{Raviv}. However, there are notable drawbacks to the utility function
approach. The notion of utility is opaque for many non-specialists, and the benefit--risk tradeoff is only implied, {\it implicitly}, through a utility function. Moreover, one can rarely obtain,
analytically, optimal policies with a general utility, hindering post-optimality analyses such as comparative statics. For instance,  \citet{Raviv} derives a differential equation satisfied by the optimal indemnity, which takes a rather complex form depending on the utility function used.

By contrast, variance, as a measure of  risk originally put forth in Markowitz's pioneering work \citep[]{Markowitz1952}, is also related to the whole distribution, yet it is more intuitive and transparent. \citet{Borch1960} designs a contract that aims to minimize the variance of the {\it insured's} liability. \citet{Kaluszka2001} extends \citet{Borch1960}'s work by incorporating a variance-related premium principle and shows that the optimal contract to minimize the variance of the insured's payment can be stop-loss, quota share
(i.e., the insurer covers a {\it constant} proportion of the loss) or a combination of the two.
\citet{Vajda1962} studies the problem from the {\it insurer's} perspective, and shows that a quota share policy minimizes the variance of indemnity in an actuarially fair contract. However, his result depends critically on limiting the admissible contracts to be such that the {\it ratio} between indemnity and loss increases as the loss increases, a feature that enables the derivation of a solution through rather simple calculus.\footnote{\citet{Vajda1962} claims that this feature ``agrees with the spirit of (re)insurance, at least in
most cases" (p. 259). However, for a larger loss, it is indeed in the spirit of insurance that the insurer should pay more, but it is not clear why he should be responsible for a higher {\it proportion}. Interestingly, our results will show that the optimal policies of our model possess this property {\it if} the insured is prudent;
see Corollary \ref{Vajda}.}

%
%

In this paper, we revisit the work of \citet{Arrow1963} by imposing a variance constraint on the {\it insurer}'s risk exposure. Unlike \citet{Vajda1962}, we consider the general actuarially unfair case and remove the restriction that the proportion of the insurer's payment increases with the size of the loss. The presence of the variance constraint causes substantial technical challenges in solving the problem. In the literature, there are  generally two approaches used to study variants of  Arrow's model: those involving sample-wise optimization and stochastic orders. However, the former fails to work for our problem due to the nonlinearity of the variance constraint, and the latter is not readily applicable either because the presence of the variance constraint invalidates the claim that any admissible contract is dominated by a stop-loss one.  The first contribution of this paper is methodological: we develop a new approach by {\it combining} the techniques of stochastic orders, calculus of variations and the Lagrangian duality to derive optimal insurance policies.
The solutions are semi-analytical in the sense that they can be computed by solving some {\it algebraic} equations (as opposed to {\it differential} equations in \citealt{Raviv}).

Because the expected value premium principle ensures the expected profit of the insurer, our model is essentially a  {\it mean--variance} model {\it \`a la} Markowitz for the {\it insurer}. 
Our second contribution is actuarial: we show that the optimal contract  is 
coinsurance above a deductible when the  variance constraint is binding. Moreover, the deductible disappears if and only if the insurance price is actuarially fair, consistent with  Mossin's Theorem \citep[]{Mossin1968}. These results are {\it qualitatively} similar to those of \citet{Raviv}, who uses a concave utility function
for the insurer. A natural question is why one would bother to study the mean--variance version of a problem that would generate contracts with similar characteristics  to its expected utility counterpart.
This question can be answered in  the same way as in the field of financial portfolio selection, where
there is an enormously large body  of study on the Markowitz mean--variance model along with its popularity in practice, {\it despite} the existence of the equally well-studied expected utility maximization models. In other words, expected utility  and mean--variance are two {\it different} frameworks, and, as argued earlier, the latter underlines a more transparent and explicit return--risk tradeoff, which usually leads to explicit solutions.

%
%
Our optimal policies involve coinsurance, which is widely utilized in the insurance industry. As pointed out by \citet{Raviv}, risk aversion on the part of the insurer could be a cause for coinsurance, but other attributes such as the nonlinearity of the insurance cost function could also lead to coinsurance.
Another explanation for coinsurance is to mitigate the  moral hazard risk; see \citet{Holmstrom1979} and \citet{Dionne1991}.   From the insured's perspective, \citet{Doherty1990} argue that default risk of the insurer can motivate the insured to choose coinsurance. \citet{Picard2000} also shows that coinsurance is optimal in order to reduce the risk premium paid to the auditor. In this paper, we prove that optimal policies  can turn from full insurance to coinsurance as the variance bound tightens, thereby
providing  a novel yet simple reason for the prevalent feature of coinsurance in insurance theory and practice: a variance bound on the insurer's risk exposure.

Intriguingly,  our optimal insurance polices {\it automatically}  satisfy the so-called \emph{incentive-compatible} condition that {\it both} the insured and the insurer pay more for a larger loss (or, equivalently,
the marginal indemnity is between 0 and 1).\footnote{The incentive-compatible condition is termed the {\it no-sabotage} condition in \citet{Carlier2003}.}
In \cite{Arrow1963}'s setting, the optimal contract -- the stop-loss one -- turns out to be incentive-compatible; however, this is generally untrue.
%
%
\citet{Gollier1996} considers an insured facing an additional background risk that is not insurable. Under the expected value principle, he discovers that the optimal insurance, which relies heavily on the dependence between the background risk and the loss, may render the marginal indemnity strictly larger than 1. \citet{Bernard:MF} generalize the insured's risk preference from expected utility to rank-dependent utility involving probability distortion (weighting), and also find that the optimal indemnity may  decrease when the loss increases. In both of these papers, the derived optimal contracts would  incentivize the insured to misreport the actual losses,
leading to ex post moral hazard. Equally absurd would be the case in which the insurer pays less for a larger loss. To address this issue, \citet{Huberman1983} propose the incentive-compatible condition as a {\it hard} constraint on  admissible insurance policies, in addition to the principle of indemnity. \citet{Xu2019} add this constraint to the model of \citet{Bernard:MF}, painstakingly developing a completely different approach in order to
overcome the difficulty arising out of this additional constraint and deriving qualitatively very different contracts. On the other hand,
\citet{Raviv} discovers that his optimal solution is incentive-compatible, assuming that the loss has a strictly positive probability density function. \citet{CarlierD2005} use a Hardy--Littlewood  rearrangement argument to prove that any optimal contract is dominated by an incentive-compatible contract, establishing the
optimality of the latter. However, their approach relies heavily on the assumption that the loss  is non-atomic.
Both of these studies rule out the important and practical case in which the loss
is atomic at $0$. By contrast, in the presence of the variance constraint, we show that the optimal policy is naturally  incentive-compatible even without the corresponding hard constraint or the assumption of an atom-less loss.

Our final contribution is a comparative statics analysis  examining the respective effects of the insured's wealth and the variance constraint on insurance demand under actuarially fair pricing. Our results indicate that the presence of a
variance bound fundamentally changes the insurance strategy -- it makes the insured's wealth relevant and it changes the way in which the two parties share the risk.
In particular, the expected coverage is
always {\it larger} for a wealthier insured who has strictly decreasing absolute prudence (DAP), rendering the insurance product a {\it normal good}. This finding provides some theoretical foundation  for the empirical observations of \citet{Millo2016} and \citet{Armantier2018}. Moreover, we show that the insurer has less downside risk  when contracting with a  wealthier insured with strictly DAP.\footnote{\citet{Menezes1980} introduce the notion of downside risk to compare two risks with the same mean and variance. The formal definition is given in Section~\ref{section:2:2}.} This result reconciles with the well-documented
phenomenon that more economically advanced regions or countries have higher insurance densities and penetrations.
On the other hand,  we establish that the variance bound significantly changes a prudent insured's risk transfer decision -- she would consistently transfer  more losses  as the variance bound loosens. A corollary of this result is, rather surprisingly, that the insurer can reduce the tail risk  by simply tightening the variance constraint. This suggests that our variance contracts do, {\it after all}, address the issue of
tail exposure.

The rest of the paper proceeds as follows. In Section \ref{modelseeting} we formulate the problem and present some preliminaries about risk preferences. In Section \ref{section:3} we develop the solution approach and present  the optimal insurance contracts. In Section \ref{section:4} we conduct a comparative analysis by examining the effects of the insured's initial wealth and the variance constraint on insurance demand. Section \ref{section:5} concludes the paper. Some auxiliary results and all proofs are relegated to the appendices.


\section{Problem Formulation and Preliminaries}\label{modelseeting}

\subsection{Variance contracts formulation}
An insured endowed with an initial wealth $w_0$ faces an insurable loss $X$, which is a non-negative, essentially bounded random variable defined on a probability space $(\Omega,\mathscr{F}, \PP)$ with  the cumulative distribution function (c.d.f.) $F_X(x):=\PP\{X\leqslant x\}$ and the essential supremum $\mathcal{M}<\infty$. An insurance contract design problem is to partition $X$ into two parts, $I(X)$ and $X-I(X)$, where $I(X)$ (the {\it indemnity}) is the portion of the loss that is ceded to the insurer (``he") and $R_I(X):=X-I(X)$ (the {\it retention}) is the portion borne by the insured (``she"). $I$ and $R_I$ are also called the insured's ceded and  retained loss functions, respectively. It is natural to require a contract to satisfy the {\it principle of indemnity}, namely the indemnity is non-negative and less than the amount of loss. Thus, the feasible set of indemnity functions is $$\mathfrak{C}:= \left\{I: 0\leqslant I(x)\leqslant x,\, \forall x\in[0,\mathcal{M}]\right\}.$$	

As the insurer covers part of the loss for the insured, he is compensated by collecting the premium from her.
Following many studies in the literature, we assume that the insurer calculates the  premium using the {\it expected value principle}. Specifically, the premium on making a non-negative random payment $Y$ is
charged as
$$
\pi(Y)=(1+\rho)\EE[Y]
$$
where $\rho\geqslant 0$ is the so-called {\it safety loading} coefficient.
His risk exposure under a contract $I$ for a loss $X$ is hence
$$
e_I(X)=I(X)-\pi(I(X)).
$$
The insurer may evaluate  this risk using  different measures  for different purposes, as \citet{Kaye2005} notes. In this paper, we assume that the insurer has sufficient regulatory capital and therefore focuses on the volatility of the underwriting risk. Specifically, he uses the variance  to measure the risk and requires
\begin{align*}
var[e_I(X)]\equiv var[I(X)]\leqslant \nu
\end{align*}
for some prescribed $\nu>0$.

On the other hand, denote by $W_I(X)$ the insured's final wealth under  contract $I$ upon its expiration, namely
$$
W_I(X)=w_0-X+I(X)-\pi(I(X)).
$$
The insured's risk preference is  characterized by a von Neumann--Morgenstern utility function $U$ satisfying $U'>0$ and $U''<0$.


Our optimal contracting problem is, therefore,
\begin{equation}\label{eqn:simplifiedP}
  \begin{array}{ll}
  \underset{I\in \mathfrak{C}}{\text{max}} & \ \ \ \EE[U(W_I(X))] \\
     \mbox{subject to} & \ \ \ var\left[e_I(X)\right]\leqslant \nu.
 \end{array}
\end{equation}
Note that this model reduces to \citet{Arrow1963}'s model
\begin{equation*} 
	\max_{I\in \mathfrak{C} } \ \ \ \EE[U(W_I(X))]
\end{equation*}
by 
setting the upper bound $\nu$ to be $\EE[X^2]$. This is because $var[e_I(X)]= var\left[I(X)\right]\leqslant \EE[X^2]$ for all $ I\in \mathfrak{C}$.

In Problem (\ref{eqn:simplifiedP}), the insured's benefit--risk consideration is captured by the utility function $U$, whereas the insurer's return--risk tradeoff is reflected by the ``mean" (the expected value principle) and the ``variance" (the variance bound). One may interpret the
problem as one faced by an insurer who likes to  design a contract with the best interest of a representative insured in mind, so as to remain marketable and competitive, while maintaining the desired profitability and variance control in the mean--variance sense.\footnote{Representative insureds in different wealth classes or different regions may have different ``typical" levels of initial wealth. Moreover, when the economy grows, a representative insured's initial wealth may change substantially. As shown in Subsection \ref{section:41}, the change in the insured's initial wealth may affect her demand for insurance.   }
Problem (\ref{eqn:simplifiedP}) can also model a tailor-made contract design for insuring a one-off event from an insured's perspective. The insured aims to maximize her expected utility while accommodating the insurer's participation constraint reflected by the mean and variance specifications.

\subsection{Absolute risk aversion and prudence}\label{section:2:2}

The Arrow--Pratt measure of absolute risk aversion \citep{Pratt1964,Arrow1965}, defined as
\begin{align*}
\mathcal{A}(x): =-\frac{U''(x)}{U'(x)},
\end{align*}
captures the dependence of the level of risk aversion on the agent's wealth $x$.
If $\mathcal{A}(x)$ is decreasing\footnote{Throughout the paper, the terms ``increasing" and ``decreasing" mean ``non-decreasing" and ``non-increasing," respectively.} in $x$, then the insured's risk preference is said to exhibit {\it decreasing absolute risk aversion} (DARA). The effect of an insured's initial wealth on the insurance demand under \cite{Arrow1963}'s model has been widely studied in the literature. It is found that a wealthier DARA insured purchases a deductible  insurance with a higher deductible. For a survey on how insureds' wealth impacts insurance, see e.g., \cite{Gollier2001,Gollier2013}.




While risk aversion ($U''<0$) captures an insured's propensity for avoiding risk, {\it prudence} (i.e., $U'''>0$) reflects her tendency to take precautions against future risk.  Many commonly used utility functions, including those with  hyperbolic absolute risk aversion (HARA) and mixed risk aversion, are prudent.\footnote{A utility function is called HARA if the reciprocal of the Arrow--Pratt measure of absolute risk aversion is a linear function, i.e., $$\mathcal{A}_U(x)=-\frac{U''(x)}{U'(x)}=\frac{1}{px+q}$$ for some $p\geqslant 0$ and $q$. It includes exponential, logarithmic and  power utility functions as special cases. For further discussion of HARA, see \citet{Gollier2001}. A utility function is said to be of mixed risk aversion if $(-1)^n U^{(n)}(x)\leqslant 0$ for all $x$ and $n=1, 2, 3, \cdots$, where $U^{(n)}$ denotes the $n$th derivative of $U$.} Based on an experiment with a large number of subjects, \citet{Noussair2014}  observe  that the majority of individuals' decisions are consistent with prudence.
\cite{Eeckhoudt1992} and \cite{Gollier1996} take into account the insured's prudence in designing optimal insurance policies. The degree of absolute prudence is defined as
\begin{align}\label{def:DAP}
	\mathcal{P}(x):=-\frac{U{'''}(x)}{U{''}(x)}
\end{align}
for a three-time  differentiable utility function $U$. If $\mathcal{P}(x)$ is strictly decreasing in $x$, then the insured is said to exhibit strictly {\it decreasing absolute prudence} (DAP).
\citet{Kimball1990} shows that DAP characterizes the notion that wealthier people are less sensitive to future risks.
Moreover, DAP implies DARA, as noted in Proposition 21 of \citet{Gollier2001}.

A  term related to prudence is {\it third-degree stochastic dominance} (TSD), which was introduced by \citet{Whitmore1970}. A non-negative random variable $Z_1$ is said to dominate another non-negative random variable $Z_2$ in TSD if
		$$
		\EE[Z_1]\geqslant \EE[Z_2]\quad\text{and}\quad\int_0^x\int_0^yF_{Z_2}(z)-F_{Z_1}(z)\dd z\dd y\geqslant 0\,\ \text{for all}\, x\geqslant 0.
		$$
		Equivalently, $Z_1$ dominates $Z_2$ in TSD if and only if $\EE[u(Z_1)]\geqslant \EE[u(Z_2)]$ for all functions $u$ satisfying $u'>0,\;u''<0$ and $u'''>0$. TSD has been widely employed  for decision making in finance and insurance. For instance,  \citet{Gotoh2000} use it to study mean-variance optimal portfolio problems.  If $Z_1$ dominates $Z_2$ in TSD and they have the same mean and variance, then  $Z_1$ is said to have less downside risk than   $Z_2$. In fact, the latter is equivalent to $\EE[u(Z_1)] \geqslant \EE[u(Z_2)]$ for any function $u$ with $u'''>0$;  see \citet{Menezes1980}.

\section{Optimal Contracts}\label{section:3}

In this section, we present our approach to solving Problem \eqref{eqn:simplifiedP}.

\medskip

First, consider Problem \eqref{eqn:simplifiedP} {\it without} the variance constraint:
\begin{equation}\label{eqn:Prob-step1}
\max_{{ I\in \mathfrak{C}}}\ \ \ \EE[U(W_I(X))].
\end{equation}
This is the classical  \citet{Arrow1963}'s model, for which the optimal contract is
a deductible one of the form $(x-d^*)_+$ for some non-negative deductible $d^*$, where $(x)_+:=\max\{x,0\}$. This contract {\it automatically} satisfies the incentive-compatible condition. Moreover, \citet{Chi2019}
(see Theorem 4.2 therein) was the first to derive an {\it analytical} form of the
 optimal deductible level $d^*$. More precisely,  define
$$
VaR_{\frac{1}{1+\rho}}(X):=\inf\left\{x\in[0,\mathcal{M}]: F_X(x)\geqslant \frac{\rho}{1+\rho}\right\}
$$
and
$$
\varphi(d):=\frac{\EE[U'(W_{(x-d)_+}(X))]}{U'(w_0-d-\pi((X-d)_+))},\quad 0\leqslant d<\mathcal{M},
$$
where 
$\inf\emptyset:=\mathcal{M}$ by convention. Then the optimal  $d^*$ is
\begin{equation}\label{eqn:d-star}
d^*=\sup\left\{VaR_{\frac{1}{1+\rho}}(X)\leqslant d<\mathcal{M}:\, \varphi(d)\geqslant  \frac{1}{1+\rho}\right\}\vee VaR_{\frac{1}{1+\rho}}(X),
\end{equation}
where $\sup\emptyset:=0$ and $x\vee y:=\max\{x,y\}$.\footnote{The number
$d^*$ can be numerically computed easily, because $\varphi(d)$ is decreasing over $[VaR_{\frac{1}{1+\rho}}(X), \mathcal{M})$; see \citet{Chi2019}.}
This leads immediately to the following proposition.
\begin{proposition}\label{proposition:31}
If $\nu\geqslant var[(X-d^*)_+]$, then $I(x)=(x-d^*)_+$ is the optimal solution to Problem \eqref{eqn:simplifiedP}.
\end{proposition}
Intuitively, if the variance bound $\nu$ is set sufficiently high, then the variance
constraint in Problem \eqref{eqn:simplifiedP} is redundant and the problem reduces to
the classical \citet{Arrow1963}'s problem. Proposition \ref{proposition:31} tells exactly and explicitly what
the bound should be for the variance constraint to be binding.

\medskip

Therefore, it suffices to solve Problem \eqref{eqn:simplifiedP} for the case in which $\nu< var[(X-d^*)_+]$, which we now set as an assumption.
\begin{Assumption}\label{assump:smallnu}
	The variance bound $\nu$ satisfies $\nu< var[(X-d^*)_+]$.
\end{Assumption}

%
%
The main thrust for finding the solution is to first restrict the analysis with a fixed level of expected indemnity and then find the optimal level of expected indemnity. To this end, we need to first identify the range in which  the optimal expected indemnity possibly lies. Noting that $var[(X-d)_+]$ is strictly decreasing and continuous in $d$ over $[{\rm ess\  inf}\;X,\mathcal{M})$, we define
$$
d_L:=\inf\{d\geqslant d^*: var[(X-d)_+]\leqslant \nu\}\quad\text{and}\quad m_L:=\EE[(X-d_L)_+].
$$
Intuitively, the insurer would  demand a deductible higher than Arrow's level $d^*$ due to the additional risk control reflected by the variance constraint, and $d_L$ is the smallest deductible that makes this
constraint binding.
\begin{lemma}\label{Lemma:m-L}
Under Assumption \ref{assump:smallnu}, for Problem \eqref{eqn:simplifiedP}, any admissible insurance policy $I$ with $\EE[I(X)]\leqslant m_L$ is no better than the deductible contract $I_L$ where $I_L(x)=(x-d_L)_+$.
\end{lemma}

Therefore, we can rule out any contract whose expected indemnity is strictly smaller than $m_L$; in other words, $m_L$ is a {\it lower} bound of the optimal expected indemnity.
In particular, no-insurance (i.e., $I^*(x)\equiv 0$) is {\it never} optimal under Assumption \ref{assump:smallnu}.

Next, we are to derive an {\it upper} bound of the optimal expected indemnity.
Consider a loss-capped contract $X\wedge k$, where $x\wedge y:=\min\{x, y\}$ and $k\geqslant 0$, which
pays the actual loss up to the cap $k$.\footnote{A loss-capped contract is also called ``full insurance up to a (policy) limit" or ``full insurance with a cap."} Define
$$
K_U:=\inf\{k\geqslant 0: var[X\wedge k]\geqslant \nu\}\quad\text{and}\quad m_U:=\EE[X\wedge K_U].
$$
In the above, $K_U$ is well-defined because $X\wedge k-\EE[X\wedge k]$ is increasing in $k$ in the sense of convex order,
according to Lemma A.2 in \citet{Chi2012}.\footnote{A random variable $Y$ is said to be greater than a random variable $Z$ in the sense of convex order, denoted as $Z\leqslant_{cx} Y$, if
$$\EE[Y]=\EE[Z]\ \text{and}\ \ \EE[(Z-d)_+] \leqslant \EE[(Y-d)_+]\;\;  \forall  d\in\mathbb{R}, $$
provided that the expectations exist. Obviously, $Z\leqslant_{cx} Y$ implies $var[Z]\leqslant var[Y]$.}
Clearly, both $K_U$ and $m_U$ depend on the variance bound $\nu$. Since $var[X]\geqslant var[(X-d^*)_+]>\nu$, we have $$K_U<\mathcal{M},\quad m_U<\EE[X]\quad\text{and}\quad var\left[X\wedge K_U\right]= \nu.$$

\begin{lemma}\label{lemma:m_U}
For any $I\in \mathfrak{C}$ with $var[I(X)]\leqslant \nu$, we must have $\EE[I(X)]\leqslant m_U$. Moreover, if $I\in \mathfrak{C}$  satisfies $var[I(X)]\leqslant \nu$ and $\EE[I(X)]=m_U$, then $I(X)=X\wedge K_U$ almost surely.
\end{lemma}

This lemma stipulates that $m_U$ is an upper bound of the optimal expected indemnity. Moreover, any admissible contract achieving this upper bound is equivalent to the
loss-capped contract $X\wedge K_U$. An immediate corollary of the lemma is
$m_L\leqslant m_U$, noting that $var[(X-d_L)_+]=\nu$.

The following result identifies the case $m_L=m_U$ as a trivial one.

\begin{proposition}\label{prop:two-point}
If $m_L=m_U$, then the loss $X$ must follow a Bernoulli distribution with values $0$ and $d_L+K_U$. Moreover, under Assumption \ref{assump:smallnu}, the optimal contract of Problem \eqref{eqn:simplifiedP} is $$I^*(0)=0\quad\text{and}\quad I^*(d_L+K_U)=K_U.$$
\end{proposition}

In what follows, we consider the general and interesting case in which $m_L< m_U$.
For $m\in(m_L, m_U)$, define
$$
\mathfrak{C}_m:=\left\{I\in \mathfrak{C}: var[I(X)]\leqslant \nu,\, \EE[I(X)]=m \right\}.
$$
We now focus on the following optimization problem
\begin{equation}\label{eqn:optimizationP}
	\max_{I\in\mathfrak{C}_m}\ \ \EE[U(W_I(X))],
\end{equation}
which is a ``cross section" of the original problem \eqref{eqn:simplifiedP} where the expected indemnity is fixed as $m$.

%

For $\lambda\in\RR$ and $\beta\geqslant 0$,
denote
\begin{align}\label{relaxed:solution}
I_{\lambda,\beta}(x):=\sup\Big\{y\in[0,x]: U'\left(w_0-x+y-(1+\rho)m\right)-\lambda-2\beta y\geqslant 0\Big\},\;x\in[0,\mathcal{M}].
\end{align}
Actually, $I_{\lambda,\beta}$ is a contract that coinsures above a deductible or coinsures following full insurance, depending on the relative values between $\lambda$ and $U'(w_0-(1+\rho)m)$. To see this, when $\lambda\geqslant U'(w_0-(1+\rho)m)$, we have
\begin{equation}\label{eqn:f-2}
I_{\lambda,\beta}(x)=\left\{
\begin{array}{ll}
0, &0\leqslant x\leqslant w_0-(1+\rho)m-(U')^{-1}(\lambda),\\
f_{\lambda,\beta}(x),& w_0-(1+\rho)m-(U')^{-1}(\lambda)< x\leqslant  \mathcal{M} ,
\end{array}
\right.
\end{equation}
and when $\lambda< U'(w_0-(1+\rho)m)$, we have
\begin{equation}\label{eqn:f-1}
I_{\lambda,\beta}(x)=\left\{
\begin{array}{ll}
x, &0\leqslant x\leqslant \frac{U'(w_0-(1+\rho)m)-\lambda}{2\beta},\\
f_{\lambda,\beta}(x),&\frac{U'(w_0-(1+\rho)m)-\lambda}{2\beta}< x\leqslant \mathcal{M},
\end{array}
\right.\end{equation}
where $f_{\lambda,\beta}(x)$ satisfies the following equation
in $y$:\footnote{It can be shown easily that this equation has a unique solution.
}
\begin{equation}\label{solution:equation}
U'(w_0-x+y-(1+\rho)m)-\lambda-2\beta y=0.
\end{equation}
Moreover, it is easy to see that $0\leqslant f_{\lambda,\beta}(x)\leqslant x$ either when 
$\lambda\geqslant U'(w_0-(1+\rho)m)$ and $w_0-(1+\rho)m-(U')^{-1}(\lambda)< x\leqslant  \mathcal{M}$, or when $\lambda< U'(w_0-(1+\rho)m)$ and $\frac{U'(w_0-(1+\rho)m)-\lambda}{2\beta}< x\leqslant \mathcal{M}$. 
Furthermore, 
\begin{align}\label{derivative:f} f'_{\lambda,\beta}(x)=\frac{-U''(w_0-x+f_{\lambda,\beta}(x)-(1+\rho)m)}{2\beta-U''(w_0-x+f_{\lambda,\beta}(x)-(1+\rho)m)}\in (0,1].
\end{align}
The following result indicates that there exists an optimal solution to Problem \eqref{eqn:optimizationP} that is in the form of $I_{\lambda,\beta}(x)$ and binds both the mean and variance constraints.

\begin{proposition}\label{Pro:optimalS}
Suppose Assumption \ref{assump:smallnu} holds and $m_L< m_U$. Then there exist $\lambda^*_m\in\RR$ and $\beta^*_m>0$ such that $I_{\lambda^*_m,\beta^*_m}$ satisfies
\begin{equation}\label{eqn11:equality-M-V}
\EE[I_{\lambda^*_m,\beta^*_m}(X)]=m\qquad\text{and}\qquad var[I_{\lambda^*_m,\beta^*_m}(X)]=\nu,
\end{equation}
and is an optimal solution to Problem \eqref{eqn:optimizationP}.
\end{proposition}

Combining Lemma~\ref{Lemma:m-L}, Lemma~\ref{lemma:m_U} and Proposition \ref{Pro:optimalS} yields that we can always find an optimal contract in one of the following three types: a deductible one of the form
$I_L(x)=(x-d_L)_+$, a loss-capped one of the form $I_U(x)=x\wedge K_U$ and a general one of the form
$I_{\lambda^*_m, \beta^*_m}(x)$. In other words, the optimal solutions of the following maximization problem
\begin{equation}\label{eqn:three-case-op}
\max_{I\in\big\{I_L, \ I_U, \ I_{\lambda^*_m, \beta^*_m}\,\text{for}\, m\in(m_L, m_U)\big\}}\EE[U(W_I(X))],
\end{equation}
where $m_L<m_U$, also solve Problem \eqref{eqn:simplifiedP}.

Note that $I_L$, $I_U$ and $I_{\lambda^*_m, \beta^*_m}$ all satisfy the incentive-compatible condition (see (\ref{derivative:f})); hence, so does at least one of the optimal contracts $I^*$ of \eqref{eqn:simplifiedP}. That is, $I^*(0)=0$ and $0 \leqslant {I^*}'(x)\leqslant 1$ almost everywhere.\footnote{As will be evident in the sequel, the values of $I'$ on a set with zero Lebesgue measure have no impact on $I$. Therefore, we will often omit  the phrase ``almost everywhere" in statements regarding the marginal indemnity function ${I}'$ throughout this paper.} Therefore, it suffices to solve the following maximization problem
\begin{equation}\label{eqn:optimal-IC}
\begin{array}{ll}
\underset{I\in \mathcal{IC}}{\text{max}} & \ \ \ \EE[U(W_I(X))] \\
\mbox{subject to} & \ \ \ var\left[e_I(X)\right]\leqslant \nu,
\end{array}
\end{equation}
where
\begin{equation}\label{eqn:IC}
\mathcal{IC}:= \left\{I: I(0)=0,\, 0 \leqslant I'(x)\leqslant 1,\,\forall x\in[0,\mathcal{M}] \right\}\subsetneq \mathfrak{C},
\end{equation}
to obtain an optimal contract for Problem \eqref{eqn:simplifiedP}.

Notice that $\mathcal{IC}$ is convex on which $\EE[U(W_I(X))]$ is  {\it strictly} concave. Using  the convex property of variance and applying arguments similar to those in the proof of Proposition 3.1 in \citet{ChiW2020}, we obtain the following proposition:

\begin{proposition}\label{exis:uniqueness}	
	\begin{itemize}
		\item[{\rm (i)}] There exist optimal solutions to Problem \eqref{eqn:optimal-IC}.
		\item[{\rm (ii)}] Assume either $\rho>0$ or $\mathbb{P}\left\{X<\epsilon\right\}>0$ for all $\epsilon>0$. Then there exists a unique solution to Problem \eqref{eqn:optimal-IC} in the sense that $I_1(X)=I_2(X)$ almost surely for any two solutions $I_1$ and $I_2$.
	\end{itemize}
\end{proposition}

Note that the assumptions in Proposition \ref{exis:uniqueness}-(ii) are satisfied in most situations of practical interest because either an insurer naturally sets a positive safety loading, or a loss actually never occurs with a positive probability, or both happen.
On the other hand, since any optimal solution to Problem \eqref{eqn:optimal-IC} also
solves Problem \eqref{eqn:simplifiedP}, Proposition \ref{exis:uniqueness}-(i) establishes the existence of optimal solutions to the latter.\footnote{It is difficult to prove the existence of solutions to Problem \eqref{eqn:simplifiedP} {\it directly}
because its feasible set is not compact only under the principle of indemnity.}
Moreover, the argument proving Proposition 3.1 in \citet{ChiW2020} can be used to
show that Proposition \ref{exis:uniqueness}-(ii) holds true for Problem \eqref{eqn:simplifiedP} as well. Finally, now that we have the existence and uniqueness of the optimal solutions for both
Problems \eqref{eqn:optimal-IC} and \eqref{eqn:simplifiedP}, we conclude that these two problems are indeed {\it equivalent} under the assumptions of
Proposition \ref{exis:uniqueness}-(ii).


While the analysis of Problem \eqref{eqn:simplifiedP} is simplified to Problem \eqref{eqn:three-case-op}, it remains challenging to solve this problem because $\lambda^*_m$ and $\beta^*_m$ are implicit functions of $m$. Before attacking this problem, we introduce a useful result that provides a general qualitative structure for the optimal indemnity  function in Problem \eqref{eqn:optimal-IC} or, equivalently, Problem \eqref{eqn:simplifiedP}.
\begin{proposition}\label{lemma:derivative}
Under Assumption \ref{assump:smallnu}, if $I^*$ is a solution to Problem \eqref{eqn:optimal-IC}, then there exists $\beta^*>0$ such that
\begin{equation}\label{NSC} {I^*}'(x)= \left\{
\begin{array}{ll}
1,&\Phi_{I^*}(x)>0,\\
c_{I^*}(x),& \Phi_{I^*}(x) = 0,\\
0,&\Phi_{I^*}(x)<0,
\end{array}\right.\end{equation}
for some function $c_{I^*}$ bounded on $[0,1]$, where
\begin{equation} \label{Phifun}
\Phi_{I}(x):=\EE\big[U'(W_{I}(X))-2\beta^*I(X)|X>x\big]-
\big((1+\rho)\EE\big[U'(W_{I}(X))\big]-2\beta^*I(X)\big),\; x\in[0,\mathcal{M})
\end{equation}
for $I\in\mathcal{IC}$.
\end{proposition}

Note that \eqref{NSC} does not entail an {\it explicit} expression of ${I^*}'$ because
its right hand side also depends on $I^*$ as well as on an unknown parameter $\beta^*$.
While deriving the optimal solution $I^*$ {\it directly} from \eqref{NSC} seems challenging,
the equation reveals the important property that  ${I^*}'$ must take a value of either 0 or 1, except at point(s) $x$ where  $\Phi_{I^*}(x) = 0$.\footnote{From the control theory perspective, \eqref{NSC} corresponds to an optimal control problem in which ${I}'$ is taken as the control variable.
Moreover, the optimal control turns out to be of the so-called ``bang-bang" type, whose values
depend on the sign of the discriminant function $\Phi_{I}$. This type of optimal control problem arises when the Hamiltonian depends linearly on control and the control is constrained between an upper bound and a lower bound. It is usually hard to  solve for optimal control when the discriminant function is complex, which is the case here.}
This property will in turn help us to decide whether the optimal contract is of the form $(x-d_L)_+, \ x\wedge K_U$, or $I_{\lambda^*_m, \beta^*_m}$.

The following theorem presents a complete solution to Problem \eqref{eqn:simplifiedP}.
\begin{theorem}\label{main:theorem}
Under  Assumption \ref{assump:smallnu} and assume that the c.d.f. $F_X$ is strictly increasing on $(0,\mathcal{M})$. We have the following conclusions:
\begin{itemize}
\item[{\rm (i)}] If $\rho=0$, then the optimal  indemnity   function is $I^*$, where  $I^*(x)$ solves the following equation in $y$
    for all $x\in (0,\mathcal{M}]$:
\begin{equation}\label{eqn:zero-rho}
U'(w_0-x+y-m^*)-2\beta^*y-U'(w_0-m^*)=0,\,\;y\in(0,x),
\end{equation}
with the parameters $m^*\in(m_L,m_U)$ and $\beta^*>0$  determined by
\begin{equation}\label{eqn:Two-moment}
\EE[I^*(X)]=m^*\quad\text{and}\quad var[I^*(X)]=\nu.
\end{equation}
\item[(ii)] If $\rho>0$, then the optimal  indemnity function is
\begin{equation}\label{eqn:I-star}
I^*(x)=\left\{
\begin{array}{ll}
0,& 0 \leqslant x\leqslant \tilde{d}, \\
f^*(x), & \tilde{d}<x\leqslant \mathcal{M},
\end{array}
\right.
\end{equation}
where $f^*(x)$ satisfies $f^*(\tilde{d})=0$ and solves the following equation in $y$:
\begin{eqnarray}\label{eqn:f-star}
&&U'(w_0-(1+\rho)m^*-x+y)-U'(w_0-(1+\rho)m^*-\tilde{d})\\
&&\quad=\frac{y}{m^*\rho}\left(U'(w_0-(1+\rho)m^*-\tilde{d})-(1+\rho)
\EE[U'(w_0-(1+\rho)m^*-X\wedge\tilde{d})]\right),\,\;y\in(0,x),\nonumber
\end{eqnarray}
and  $\tilde{d}\in (VaR_{\frac{1}{1+\rho}}(X),\mathcal{M})$ and $m^*\in(m_L, m_U)$ are determined by \eqref{eqn:Two-moment}.
 \end{itemize}
\end{theorem}

Theorem \ref{main:theorem} provides a complete solution  to Problem \eqref{eqn:simplifiedP}.
It indicates that the optimal contract can {\it not} be a pure deductible of the form $(x-d_L)_+$, nor a pure loss-capped of the form $x\wedge K_U$. It can only be in the form $I_{\lambda^*_m, \beta^*_m}$ of \eqref{eqn:f-2} (rather than \eqref{eqn:f-1}). The optimal policies can be computed by solving a system of three {\it algebraic} equations; so the result is semi-analytic.


Actuarially, Theorem \ref{main:theorem} reveals how the variance bound  impacts
the contract. When the bound $\nu$ is sufficiently low so that it is binding (hence the model does not degenerate into
the classical \citealt{Arrow1963}'s model), the optimal policy is always genuine coinsurance
if there is no safety loading. Here, by ``genuine" we mean the {\it strict} inequalities
$0<I^*(x)<x$ for all $x\in (0,\mathcal{M}]$, namely
both the insurer and the insured pay positive portions  of the loss incurred. When the safety loading coefficient is positive, the optimal contract demands genuine coinsurance above a positive deductible. So the variance bound translates into a change from the part of the full insurance in Arrow's contract to coinsurance.
Our contracts are similar {\it qualitatively} to those of \citet{Raviv}, in which a utility function is in the place of a variance bound; however, ours  are {\it quantitatively} different from \citet{Raviv}'s.

On the other hand, the deductible $\tilde{d}$ is positive if and only if the safety loading coefficient is positive. So the existence of the deductible is completely determined by the  loading coefficient in the insurance premium. This result is consistent with Mossin's Theorem (\citealt{Mossin1968}).

\begin{corollary}\label{Vajda}
Under the assumptions of Theorem \ref{main:theorem}, if the insured is prudent, then the proportion between optimal indemnity and loss increases as loss increases.
\end{corollary}

So, with a prudent insured, the insurer pays more not only absolutely  but also relatively as loss increases. \citet{Vajda1962} restricts his study on a variance contracting problem to policies that have this feature of the insurer covering proportionally more for larger losses. Corollary \ref{Vajda} uncovers ex post this feature in our optimal policies, {\it provided} that the insured is prudent.

\section{Comparative Statics}\label{section:4}

Thanks to the semi-analytic results derived in the previous section, we are able to
analyze the impacts of the insured's initial wealth and the variance bound on the insurance demand.

We make the following assumptions for our comparative statics analysis:
\begin{Assumption}\label{Assumption}
\begin{itemize}
\item[{\rm (i)}] $F_X$ is strictly increasing on $(0,\mathcal{M})$.
\item[{\rm (ii)}] The insurance is fairly priced, i.e., $\rho=0$.
\end{itemize}
\end{Assumption}

Assumption \ref{Assumption}-(i) is standard in the literature that accommodates most of the used distributions by actuaries, such as exponential, lognormal, gamma, and Pareto distributions. Assumption \ref{Assumption}-(ii) is not necessarily plausible in practice,  but it is
 meaningful in theory, as it describes a state in competitive equilibrium, in which insurers break even and  insurance policies are actuarially fair for representative insureds \citep[see e.g.,][]{Rothschild1976,Viscusi1979}. It is important to carry out comparative statics
analyses in such a ``fair" state in order to rule out any impact emanating  from an unfair
price.
Such an assumption is indeed often imposed when conducting comparative statics in the literature of insurance economics. For example, the comparative statics results of \citet{Ehrlich1972} and \citet{Viscusi1979} deal exclusively with actuarially fair situations. Many recent studies, such as \citet{Eeckhoudt2003}, \citet{Huang2006} and \citet{Teh2017}, also impose this assumption for their comparative statics analyses.

Finally, we will assume $\nu<var[X]$ throughout this section, as otherwise the variance constraint is redundant and  the optimal solution is trivially full insurance.

\subsection{Impact of the insured's initial wealth}\label{section:41}

In this subsection we examine the impact of the insured's initial wealth on insurance demand. We first recall the notion of one function up-crossing another.
A function $g_1$ is said to {\it up-cross} a function $g_2$, both defined on $\RR$,  if there exists $z_0\in\RR$ such that
			\[\left\{\begin{array}{ll}
			g_1(x)\leqslant g_2(x),& \ x<z_0,\\
			g_1(x)\geqslant g_2(x),& \ x\geqslant z_0.\end{array}\right.\]
Moreover,  $g_1$ is said to up-cross $g_2$ twice if there exist $z_0<z_1$ such that
		\[\left\{\begin{array}{ll}
		g_1(x)\leqslant g_2(x),& \ x< z_0,\\
		g_1(x)\geqslant g_2(x),& \ z_0\leqslant x<z_1,\\
		g_1(x)\leqslant g_2(x),& \ x\geqslant z_1.\end{array}\right.\]

Consider two initial wealth levels $w_1<w_2$ and denote the corresponding optimal contracts by $I_1^*$ and $I_2^*$ and the associated parameters by $\beta^*_1$ and $\beta_2^*$, respectively, which are determined by Theorem \ref{main:theorem}. Recall that $\rho=0$; so the insurer's risk exposure functions are
\begin{equation}\label{eqn:e-I}
e_{I_i^*}(x)=I_i^*(x)-\EE[I_i^*(X)]=I_i^*(x)-m_i^*,\;\;i=1,2,
\end{equation}
where $m_i^*:=\EE[I_i^*(X)]$. Taking  expectations on \eqref{eqn:zero-rho} yields $$U'(w_i-m_i^*)=\EE[U'(w_i-X+e_{I_i^*}(X))]-2\beta_i^*m_i^*,$$ which in turn implies,  for $i=1,2$,
\begin{align}\label{tilde:equation}
U'(w_i-x+e_{I_i^*}(x))-2\beta_{i}^*e_{I_i^*}(x)-\EE[U'(w_i-X+e_{I_i^*}(X))]=0,
\end{align}
\begin{align}\label{mean:variance}
\EE[e_{I_i^*}(X)]=0 \quad\text{and} \ \ \ \EE[(e_{I_i^*}(X))^2]=var[I_i^*(X)]=\nu.
\end{align}
Note that the insurer's profit with the contract $I_i$ is
$\EE[I_i^*(X)]-I_i^*(X)\equiv -e_{I_i^*}(X)$, $i=1,2$.
The following theorem establishes the impact of the initial wealth on the  insurance contract.

\begin{theorem}\label{CS:initialwealth}
In addition to Assumption~\ref{Assumption}, we assume that $\nu<var[X]$ and the insured's utility function $U$  exhibits strictly DAP. Then, the insurer's risk exposure function with the larger initial wealth, $e_{I_2^*}(x)$, up-crosses the risk exposure function with the smaller initial wealth, $e_{I_1^*}(x)$, twice.
 Moreover,  the insurer's profit, $-e_{I_2^*}(X)$,  has less downside risk
when contracting with the wealthier insured.
\end{theorem}

\begin{figure}[htp]\label{Theo41}
	\centering
	\includegraphics[width=5 in]{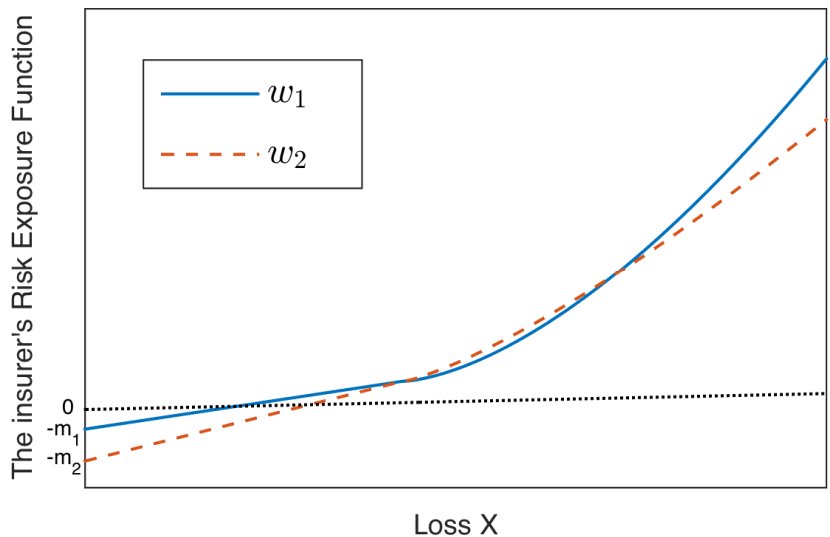}
	\caption{Comparison of two insurer's risk exposure functions with $w_1<w_2$.}\label{figure:Theo41}
\end{figure}

Figure \ref{figure:Theo41} illustrates graphically the first part of  Theorem \ref{CS:initialwealth}. The actuarial implication is that when the insured becomes wealthier, the insurer's risk exposure is lower  for  large or  small losses and is higher   for  moderate losses. This can be explained  intuitively  as follows. Even if  the insurance pricing is actuarially fair, the  insureds are unable to transfer all the risk to the insurer due to the variance bound. However,  the wealthier insured is more tolerant with large losses due to the DAP; hence, the insurer's risk exposure is lower for large losses when contracting with the wealthier insured. Due to the requirement that the insurer's {\it expected} risk exposure be always zero, the insurer's risk exposure with the wealthier insured must be higher for moderate losses. Now, should  the insurer's risk exposure with the wealthier insured also be higher for small losses, then overall the insurer's risk exposure with the wealthier insured would be strictly more spread out than that with the less wealthy one, leading to a smaller variance of the former, which would be a contradiction. Hence, the insurer's risk exposure must be {\it lower} for small  losses  with the wealthier insured.

The second part of the theorem, on the other hand, suggests that  a variance minding  insurer prefers to provide insurance to a wealthier insured due to the smaller downside risk. Such a finding may shed light on why insurers underwrite relatively more business in developed countries or, in a same country,
engage more business when the economy improves.\footnote{For example, \cite{Hofmann2015}, an industry report from the   insurance company Zurich,
shows that both {\it insurance densities} (premiums per capita) and {\it insurance penetrations} (premiums as a percent of GDP) of advanced economies are much higher than those of emerging economies. This report also demonstrates that insurance markets in both advanced and emerging economies experience rapid growth when the economies grow.}

\begin{corollary}\label{initialwealth:betaandm}
	Under the assumptions of Theorem \ref{CS:initialwealth}, we have the following conclusions:
	\begin{itemize}
		\item[{\rm (i)}] $\EE[I_1^*(X)]<\EE[I_2^*(X)]$ and $\beta_2^*<\beta_1^*$.
		\item[{\rm (ii)}] Either $I^*_{1}(x)<I^*_{2}(x)$ $\forall x> 0$, or $I^*_{1}$ up-crosses $I^*_{2}$.
	\end{itemize}

\end{corollary}

\begin{figure}[htp]\label{Coro42}
	\centering
	\includegraphics[width=4 in]{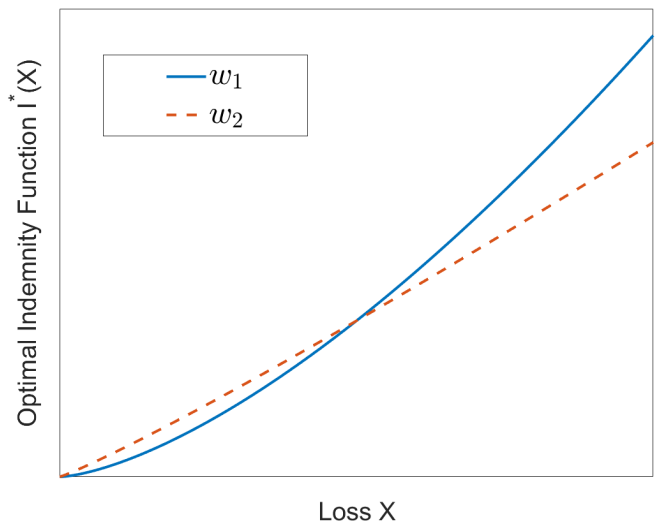}
	\caption{Comparison of two optimal indemnity functions with $w_1<w_2$.}\label{figure:Coro42}
\end{figure}

In  part (ii) of this corollary, while the case $I^*_{1}(x)<I^*_{2}(x)$ $\forall x> 0$ is a special case of  $I^*_{1}$ up-crossing $I^*_{2}$, we state it separately to highlight its possibility. In the classical \cite{Arrow1963}'s model, full insurance is optimal when insurance pricing is actuarially fair. This conclusion is {\it independent} of an insured's worth. \citet{Zhou2010} show that such a conclusion
is intact when there is an {\it exogenous} upper limit imposed on the insurer's risk exposure. Our result yields that adding a variance bound fundamentally changes the insurance demand -- it makes the insured's wealth level relevant, and it changes the way in which the two parties share the risk.  Specifically, Corollary \ref{initialwealth:betaandm} suggests that a DAP wealthier insured would either demand more coverage across the board or retain more larger risk and cede more smaller risk (see Figure \ref{figure:Coro42}). Either way, the {\it expected} coverage is {\it always} larger for the wealthier insured. Recall that  insurance is called a {\it normal (inferior) good} if wealthier people purchase more (less) insurance converge; see \citet{Mossin1968}, \citet{Schlesinger1981} and \citet{Gollier2001}.
\citet{Millo2016} argues that nonlife insurance is a normal good by empirically testing whether or not income elasticity is significantly greater than one. \citet{Armantier2018} use micro level survey data on households' insurance coverage to conclude that insurance is a normal good, thereby  providing a better understanding of the relationship between insurance demand and economic development.
These studies, however, are purely empirical. To the best of our knowledge, ours is the first
{\it theoretical} result regarding insurance as a normal good under the insurer's variance constraint, confirming these empirical findings.\footnote{\citet{Mossin1968}, \citet{Schlesinger1981} and \citet{Gollier2001} show that a wealthier insured with a DARA preference will cede less risk under {\it unfair} insurance pricing; hence, insurance is an inferior good in the corresponding economy. Their results degenerate into full insurance when the pricing is fair, and thus insurance demand  is independent of the insured's wealth. Consequently, our results do {\it not} contradict theirs. }

\subsection{Impact of the variance bound}

In this subsection we keep the insured's initial wealth unchanged and analyze the impact of the variance bound on her  demand for insurance. Consider two variance bounds with $0<\nu_1<\nu_2<var[X]$ and denote the corresponding optimal indemnity functions by $I_1^*$ and $I_2^*$ and the parameters by $\beta^*_1$ and $\beta_2^*$, respectively. Thus, the insurer's risk exposures,  $e_{I_i^*}(x)=I_i^*(x)-\EE[I_i^*(X)]$, $i=1,2,$ satisfy
\begin{align}\label{tilde:equation:2}
U'(w_0-x+e_{I_i^*}(x))-2\beta_{i}^*e_{I_i^*}(x)-\EE[U'(w_0-X+e_{I_i^*}(X))]=0
\end{align}
and
\begin{align}\label{mean:variance:2}
\EE[e_{I_i^*}(X)]=0 \quad\text{and} \ \ \ \EE[(e_{I_i^*}(X))^2]=var[e_{I_i^*}(X)]=\nu_i.
\end{align}
The following theorem illustrates how the insurer's risk exposure responds to the change in the variance bound.
\begin{theorem}\label{CS:variance}
Under Assumption~\ref{Assumption},  the insurer's risk exposure function with the larger variance bound, $e_{I_2^*}$, up-crosses that with the smaller variance bound, $e_{I_1^*}$.
\end{theorem}

Under fair insurance pricing, this theorem indicates that, as the variance bound decreases, the insurer is exposed to less risk  for a larger $X$ and to more risk for a smaller  $X$. This result has a rather significant implication in terms of the insurer's tail risk management. A variance constraint by its very definition does not control the tail risk {\it directly}. However, Theorem \ref{CS:variance} suggests that the insurer can reduce the risk exposure for larger losses simply by tightening the variance constraint.\footnote{It follows from Lemma \ref{convex:order} that the insurer with a more relaxed  variance constrant suffers more underwriting risk in the sense of convex order, i.e., $e_{I^*_1}(X)\leqslant_{cx}e_{I^*_2}(X)$.} This further justifies our formulation of the variance contracting model.


\begin{corollary}\label{Coro1:variance}
Under the assumption of Theorem \ref{CS:variance}, for any  $0<\nu_1<\nu_2<var[X]$, we have the following conclusions:
\begin{itemize}
    \item[{\rm (i)}] $\EE[I^*_1(X)]<\EE[I^*_2(X)]$ and $\beta_2^*<\beta_1^*$;
    \item[{\rm (ii)}] If the insured's utility function satisfies $U'''\geqslant 0$, then $I_1^*(x)<I_2^*(x)$ $\forall x> 0.$
\end{itemize}
\end{corollary}

\begin{figure}[htp]\label{Coro44}
	\centering
	\includegraphics[width=4 in]{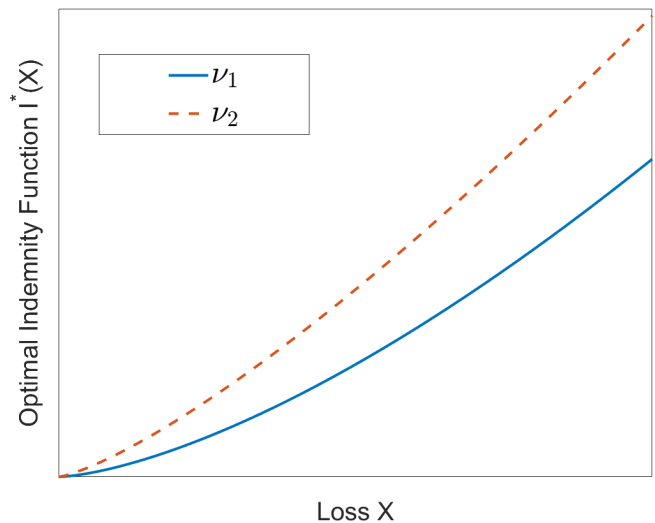}
	\caption{Comparison of two optimal indemnity functions with $\nu_1<\nu_2$.}\label{figure:Coro44}
\end{figure}


Corollary \ref{Coro1:variance}-(i) can be easily interpreted: an insurer with a tighter variance bound offers less expected coverage.
As a complement to Theorem \ref{CS:variance}, Corollary \ref{Coro1:variance}-(ii) establishes  a  direct characterization of the insured's optimal risk transfer  with regard to  the change in the variance  bound: A prudent  insured  consistently cedes more losses  when the  variance bound increases  (see Figure \ref{figure:Coro44}). In other words, if the insurance contract is priced fairly, the insured will transfer as much risk to the insurer as the latter's  risk tolerance allows.


\section{Concluding Remarks}\label{section:5}

In this paper, we have revisited the classical \cite{Arrow1963}'s model by adding a variance limit on the insurer's  risk exposure.
This constraint is motivated by the insurer's desire to manage underwriting risk; at the same time, it poses considerable technical challenges for solving the problem.
We have developed an approach to derive optimal contracts semi-analytically,
  in the form of coinsurance above a deductible when the variance constraint is active. The final policies automatically satisfy the incentive-compatible condition, thereby eliminating potential ex post moral hazard.
We have also conducted a comparative analysis to examine the impact of the insured's wealth and of the variance bound on insurance demand.

This work can be extended in a couple of directions. First, we have restricted the comparative analysis to actuarially fair insurance. Analyzing for  the general unfair case calls for a different approach than the one presented here. Second, a model incorporating probability distortion (weighting) is of significant interest, both theoretically and practically. This is because probability distortion, a phenomenon well documented in psychology and behavioral economics, is related to tail events, about which both insurers and insureds have great concerns.

\bigskip

\bigskip

\noindent{\bf\Large Appendices}

\appendix\label{appendix}
\section{Stochastic Orders}
Since the notion of stochastic orders	plays an important role in this paper, we present, in this appendix, some useful results in this regard.
	
	A random variable $Y$ is said to be greater than a random variable $Z$ in the sense of stop-loss order, denoted as $Z\leqslant_{sl} Y$, if
	$$\EE[(Z-d)_+] \leqslant \EE[(Y-d)_+]\;\;  \forall  d\in\mathbb{R}, $$
	provided that the expectations exist. It follows readily that $Y$ is greater than $Z$ in convex order (i.e., $Z\leqslant_{cx} Y $), if $\EE[Y]=\EE[Z]$ and $Z\leqslant_{sl} Y$.
	
	A useful way to verify the stop-loss order is the well--known Karlin--Novikoff criterion (\citealt{Karlin1963}).
	\begin{lemma}\label{Karlin:Novikoff}
		Suppose $\EE[Z]\leqslant \EE[Y]<\infty$. If $F_Z$ up-crosses $F_Y$, then $Z\leqslant_{sl} Y$.
	\end{lemma}
	If $Z\leqslant_{sl} Y$, then $\EE[g(Z)]\leqslant \EE[g(Y)]$ holds for all the increasing convex functions $g$, provided that the expectations exist. Based on the Karlin--Novikoff criterion, \citet{GollierSchlesinger1996} obtain the following lemma.

\begin{lemma}\label{lemma:convex:d}
	For any $h \in \mathfrak{C}$, we have $X\wedge d\leqslant_{cx}h(X)$,  where $d\in[0,\mathcal{M}]$ satisfies $\EE[X\wedge d]=\EE[h(X)]$.
\end{lemma}

The following result with respect to convex order is from Lemma 3 of \citet{Ohlin1969}.
\begin{lemma}\label{convex:order}
	Let $Y$ be a random variable and $h_i,\; i = 1, 2, $ be two increasing functions with
	$\EE[h_1(Y )] = \EE[h_2(Y )]$. If $h_1$ up-crosses $h_2$, then $h_2(Y ) \leqslant_{cx} h_1(Y )$.
\end{lemma}

\section{Other Useful Lemmas}
This appendix presents some other technical results that are useful in connection with this paper.
%
%
%

It is easy to verify that any sequence of  indemnity  functions in $\mathcal{IC}$ is uniformly bounded and equicontinuous over $[0,\mathcal{M}]$. Hence, the Arz\'{e}la-Ascoli theorem implies
%
\begin{lemma}\label{compact}
	The set $\mathcal{IC}$ is compact under the norm $d(I_1, I_2)=\max_{t\in [0,\mathcal{M}]} \mid I_1(t)-I_2(t) \mid$, $I_1, I_2\in \mathcal{IC}$.
\end{lemma}


For the following lemma, one can refer to \citet{Komiya1988} for a proof.
\begin{lemma}  \label{sion}
	(Sion's Minimax Theorem) \ Let Y be a compact convex subset of a linear topological space and Z a convex subset of a linear topological space. If $\Gamma$ is a real-valued function on $Y\times Z$ such that $\Gamma(y,\cdot)$ is continuous and concave on $Z$ for any $y\in Y$ and $\Gamma(\cdot,z)$ is continuous and convex on $Y$ for any $ z\in Z$, then $$\min\limits_{{y\in Y}}  \  \max\limits_{{z\in Z}} \Gamma(y,z)=\max\limits_{{z\in Z}} \ \min\limits_{{y\in Y}} \Gamma(y,z).$$
\end{lemma}


The following lemmas are needed in the comparative analysis.
\begin{lemma}\label{lemma:upcross:1}
	If a non-negative increasing function $h_1$ up-crosses a non-negative increasing function $h_2$ with $\EE[(h_1(X))^2]=\EE[(h_2(X))^2]$, then either $h_1(X)$ and $h_2(X)$ have the same distribution or $\EE[h_1(X)]<\EE[h_2(X)]$.
\end{lemma}
\begin{proof}
	If $\EE[h_2(X)]\leqslant \EE[h_1(X)]$,  then Lemma \ref{Karlin:Novikoff} implies
	$h_2(X)\leqslant_{sl}h_1(X)$.
	Moreover, we have
	$$
	\EE[(h_i(X))^2]=2 \int_0^{\infty}\EE[(h_i(X)-t)_+]\dd t.
	$$
	Since it is assumed that $\EE[(h_1(X))^2]=\EE[(h_2(X))^2]$, we must have $\EE[(h_1(X)-t)_+]=\EE[(h_2(X)-t)_+]$ for any $t\geqslant 0$. It then follows from the equation  $\EE[(h_i(X)-t)_+]=\int_t^{\infty}(1-F_{h_i(X)}(y))\dd y$ that $h_1(X)$ and $h_2(X)$ have the same distribution.
\end{proof}

\begin{lemma}\label{distribution:same}
Under the assumption of Theorem \ref{CS:initialwealth}, it is impossible that $e_{I_1^*}(X)$ and $e_{I_2^*}(X)$ have the same distribution,  and it is impossible either $e_{I_1^*}$ up-crosses $e_{I_2^*}$ or $e_{I_2^*}$ up-crosses $e_{I_1^*}$, where $e_{I_i^*}$ is given in \eqref{eqn:e-I}.
\end{lemma}
 \begin{proof}
	 First of all, we show that $e_{I_1^*}(X)$ and $e_{I_2^*}(X)$ cannot have the same distribution. Define
	\begin{align}\label{phi}
		\phi(z):=\frac{-U''(w_1-z)}{-U''(w_2-z)},\,  \ 0\leqslant z<w_1<w_2.
	\end{align}
	A direct calculation based on the assumption of strictly DAP  shows that
	$$
	\phi'(z)=(\mathcal{P}(w_1-z)-\mathcal{P}(w_2-z))\phi(z)>0,
	$$
	where $\mathcal{P}$ is defined in \eqref{def:DAP}. As a result, $\phi$ is a strictly increasing function.
	
	We now prove the result by contradiction. Assume that $e_{I_1^*}(X)$ and $e_{I_2^*}(X)$ are equal in distribution. Noting that $e_{I_i^*}$ is increasing and Lipschitz-continuous and that $F_X(x)$ is strictly increasing, we have
	$$
	e_{I_1^*}(x)=e_{I_2^*}(x),\,\forall x\in[0,\mathcal{M}],
	$$
	which in turn implies $e_{I_1^*}'(x)= e_{I_2^*}'(x).$ It follows from \eqref{tilde:equation} that
	$e_{I_i^*}'(x) =\frac{1}{1+2\beta_i^*/(-U''(w_i-x+e_{I_i^*}(x))}$ for all $x\geqslant 0$; hence,
		$$
	\frac{2\beta_1^*}{-U''(w_1-x+e_{I_1^*}(x))}=\frac{2\beta_2^*}{-U''(w_2-x+e_{I_2^*}(x))}.
	$$
	Because $e_{I_1^*}(x)=e_{I_2^*}(x)$ for all $x\in[0,\mathcal{M}]$, we obtain
	$$
	\phi(x-e_{I_1^*}(x))=\frac{-U''(w_1-x+e_{I_1^*}(x))}{-U''(w_2-x+e_{I_1^*}(x))}=\frac{\beta_1^*}{\beta_2^*},\,\forall x\in [0,\mathcal{M}],
	$$
	which contradicts the fact that $x-e_{I_1^*}(x)\equiv x-I_1^*(x)+\EE[I_1^*(X)]$ is strictly increasing in $x\in [0,\mathcal{M}]$ and that $\phi$ is a strictly increasing function.

	Next we show that it is impossible that $e_{I_1^*}$ up-crosses $e_{I_2^*}$. Again we prove the result by contradiction. Since $e_{I_1^*}$ is not always non-negative, we introduce the  increasing function $$\tilde{f}_i^*(x):=e_{I_i^*}(x)+\widetilde{m},\;\;i=1,2,$$ where $\widetilde{m}:=\max\left\{\EE[I_1^*(X)],\EE[I_2^*(X)]\right\}$. It then follows that $\tilde{f}_1^*$ up-crosses $\tilde{f}_2^*$, and
	$$\tilde{f}_i^*(x)\geqslant 0,\quad \EE[\tilde{f}_i^*(X)]=\widetilde{m},\quad \EE[(\tilde{f}_i^*(X))^2]=\EE[(e_{I_i^*}(X))^2]+\widetilde{m}^2=\nu+\widetilde{m}^2,\;\;i=1,2,$$ where the second equality follows from the fact that  $\EE[e_{I_i^*}(X)]=0$. Lemma \ref{lemma:upcross:1} yields that $\tilde{f}_1^*(X)$ and $\tilde{f}_2^*(X)$ have the same distribution, and hence so do $e_{I_1^*}(X)$ and $e_{I_2^*}(X)$. As shown above, $e_{I_1^*}(X)$ and $e_{I_2^*}(X)$ cannot have the same distribution, and therefore $e_{I_1^*}$ cannot up-cross $e_{I_2^*}$. A similar analysis shows that it is impossible that $e_{I_2^*}$ up-crosses $e_{I_1^*}$. The proof is thus complete.
\end{proof}

\section{Proofs}\label{appendix:B}

\noindent
\textbf{Proof of Lemma \ref{Lemma:m-L}:} \\

For any $I\in \mathfrak{C}$ with $\EE[I(X)]\leqslant m_L$, it follows from Lemma \ref{lemma:convex:d} that $$X\wedge d_I\leqslant_{cx}R_I(X),$$ where $d_I\geqslant d_L$ is determined by $\EE[(X-d_I)_+]=\EE[I(X)]$  (or, equivalently, $\EE[X\wedge d_I]=\EE[R_I(X)]$). Thus, we have $\EE[U(W_I(X))]\leqslant \EE[U(W_{(x-d_I)_+}(X))]$. Furthermore, according to the proof of Theorem 4.2 in \citet{Chi2019}, $\EE[U(W_{(x-d)_+}(X))]$ is a decreasing function of $d$ over $[d^*, \mathcal{M})$, where $d^*$ is defined in \eqref{eqn:d-star}. Recalling that $d_L\geqslant d^*$, we conclude that $I$ is no better than $(x-d_L)_+$.\\

\noindent
\textbf{Proof of Lemma \ref{lemma:m_U}:} \\

We prove by contradiction. Assume $\EE[I(X)]> m_U$ for some indemnity  function $I\in\mathfrak{C}$ satisfying $var[I(X)]\leqslant \nu$. Lemma \ref{lemma:convex:d} implies that there  exists $K>K_U$ such that $$X\wedge K\leqslant_{cx}I(X).$$ Noting that $var[X\wedge K]$ is strictly increasing and continuous in $K\in [K_U,\mathcal{M}]$, we obtain $$\nu=var[X\wedge K_U]<var[X\wedge K]\leqslant var[I(X)]\leqslant \nu,$$ leading to a contradiction.
	
For any $I\in \mathfrak{C}$ satisfying $var[I(X)]\leqslant \nu$ and $\EE[I(X)]=m_U$, it follows from Lemma \ref{lemma:convex:d} that
$$
X\wedge K_U\leqslant_{cx}I(X),
$$
which in turn implies $$\nu= var[X\wedge K_U]\leqslant var[I(X)]\leqslant \nu.$$ Because  
$$
var[Z]=2\int_{-\infty}^{\infty}\big\{\EE[(Z-t)_+]-(\EE[Z]-t)_+\big\}\dd t
$$
for any random variable $Z$ with a finite second moment, we deduce  $\EE[(X\wedge K_U-t)_+]=\EE[(I(X)-t)_+]$ almost everywhere in $t$. Therefore, $X\wedge K_U$ and $I(X)$ are equally distributed, which implies  $\mathbb{P}(I(X)\leqslant K_U)=1$. Moreover, it follows from $I\in \mathfrak{C}$ that $\mathbb{P}(I(X)\leqslant X\wedge K_U)=1$. Since $\EE[I(X)]=\EE[X\wedge K_U]$, we obtain that $I(X)=X\wedge K_U$ almost surely. The proof is thus complete. \\
%


\noindent
\textbf{Proof of Proposition~\ref{prop:two-point}:} \\

Because $m_L=m_U$, it follows from Lemma \ref{lemma:m_U} that $X\wedge K_U=(X-d_L)_+ $ almost surely. Thus, it follows from the fact that $d_L>0$ that $X$ must follow a Bernoulli distribution with values $0$ and $d_L+K_U$. Furthermore, Lemmas~\ref{Lemma:m-L} and~\ref{lemma:m_U} imply that any admissible insurance policy is no better than $(x-d_L)_+$. Therefore, the optimal indemnity must be $K_U$ at point $d_L+K_U$.\\


%

\noindent
\textbf{Proof of Proposition \ref{Pro:optimalS}:} \\

Introducing two Lagrangian multipliers $\lambda\in\RR$ and $\beta\geqslant 0$, we consider the following maximization problem:
\begin{equation}\label{prob:lagr:2}
\max_{I\in\mathfrak{C}} U_{\lambda,\beta}(I):=\EE\Big[ U(w_0-X+I(X)-(1+\rho)m)-\lambda (\EE[I(X)]-m)-\beta (\EE[I^2(X)]- \nu-m^2)\Big].
\end{equation}
Fix $x\geqslant 0$. The above objective function motivates the introduction of the following function:
$$
\Psi(y):=U(w_0-x+y-(1+\rho)m)-\lambda y-\beta y^2,\,0\leqslant y\leqslant x.
$$
The assumption of $U''<0$ implies that $\Psi$ is strictly concave  with  $$\Psi'(y)=U'(w_0-x+y-(1+\rho)m)-\lambda -2\beta y.$$
As a consequence, for each $x\geqslant 0$, $I_{\lambda,\beta}(x)$ defined in \eqref{relaxed:solution} is an optimal solution to $$\max_{0\leqslant y\leqslant x}\Psi(y).$$ This result, together with the fact that $I_{\lambda,\beta}\in \mathfrak{C}$, implies that  $I_{\lambda,\beta}$ solves Problem \eqref{prob:lagr:2}.

Notably, if there exist $\lambda^*_m\in \RR$ and $\beta^*_m> 0$ such that
\begin{equation}\label{eqn:equality-M-V}
\EE[I_{\lambda^*_m,\beta^*_m}(X)]=m\qquad\text{and}\qquad var[I_{\lambda^*_m,\beta^*_m}(X)]=\nu,
\end{equation}
then  $I_{\lambda^*_m,\beta^*_m}$ solves Problem \eqref{eqn:optimizationP}.  We prove this  by contradiction. Indeed, if there exists $I_*\in \mathfrak{C}_m$ such that $\EE[U(W_{I_{*}}(X))]>\EE[U(W_{I_{\lambda^*_m,\beta^*_m}}(X))]$, then we can obtain $$U_{\lambda^*_m,\beta^*_m}(I_*)>U_{\lambda^*_m,\beta^*_m}(I_{\lambda^*_m,\beta^*_m}),$$ which contradicts the fact that $I_{\lambda^*_m,\beta^*_m}$ solves Problem \eqref{prob:lagr:2} with $\lambda=\lambda^*_m$ and $\beta=\beta^*_m.$ Therefore, we only need to show the existence of $\lambda^*_m$ and $\beta^*_m$.

It follows from \eqref{eqn:f-2}-\eqref{derivative:f} that $I_{\lambda,\beta}$ satisfies the incentive-compatible condition, i.e., $I_{\lambda,\beta}\in \mathcal{IC}$. Such an observation motivates us to consider an auxiliary problem
	
	\begin{equation}\label{eqn:optimizationP:auxiliary}
	\max_{I\in\mathcal{IC}_m}\ \ \EE[U(W_I(X))],
	\end{equation}	
	where $$\mathcal{IC}_m:=\left\{I(x)\in \mathcal{IC}: var[I(X)]\leqslant \nu,\, \EE[I(X)]=m \right\}.$$ Problem \eqref{eqn:optimizationP:auxiliary} differs from  Problem \eqref{eqn:optimizationP} in that the feasible set is $\mathcal{IC}_m$ instead of $\mathfrak{C}_m$. 
		In what follows, we show that $\mathcal{IC}_m\neq\emptyset$, and that there exists a unique optimal solution $I^*_m$ to Problem \eqref{eqn:optimizationP:auxiliary} satisfying $var[I^*_m(X)]=\nu$. Indeed, for any $m\in(m_L, m_U)$, let $\theta\in(0,1)$ be such that $m=\theta m_U+(1-\theta)m_L$. Denote
		$$
		I_{\theta}(x):=\theta (x\wedge K_U)+(1-\theta)(x-d_L)_+.
		$$
		It follows that $\EE[I_{\theta}(X)]=m$ and $$\sqrt{var[I_{\theta}(X)]}\leqslant \theta \sqrt{var[X\wedge K_U]}+(1-\theta)\sqrt{var[(X-d_L)_+]}=\sqrt{\nu}.$$ As a consequence,  $I_{\theta}\in \mathcal{IC}_m$; and hence $\mathcal{IC}_m$ is nonempty. Moreover, note that there must exist  $d_m\in[0, d_L)$ such that $\EE[(X-d_m)_+]=m$. For any $I\in \mathcal{IC}_m$,  define
		$$
		\tilde{I}_{\alpha}(x):=\alpha I(x)+(1-\alpha)(x-d_m)_+,\; \alpha\in[0,1].
		$$
		Then, we have $\EE[\tilde{I}_{\alpha}(X)]=m$ and $$var[\tilde{I}_1(X)]=var[I(X)]\leqslant \nu=var[(X-d_L)_+]\leqslant var[(X-d_m)_+]=var[\tilde{I}_0(X)],$$
		where the last inequality follows from  the fact that $var[(X-d)_+]$ is decreasing in $d$.  Because $var[\tilde{I}_{\alpha}(X)]$ is continuous in $\alpha$, there must exist  $\alpha^{*}\in[0,1]$ such that $var[\tilde{I}_{\alpha^{*}}(X)]=\nu$. Lemma \ref{lemma:convex:d} yields that $$X\wedge d_m\leqslant_{cx}X-I(X),$$ leading to $$\EE[U(W_{\tilde{I}_{\alpha^*}}(X))]\geqslant \alpha^*\EE[U(W_I(X))]+(1-\alpha^*)\EE[U(W_{(X-d_m)_+}(X))]\geqslant \EE[U(W_I(X))].$$  This means that $I$ is no better than $\tilde{I}_{\alpha^{*}}$. Together with the Arz\'{e}la-Ascoli theorem, the above analysis  implies that there exists an optimal solution to Problem \eqref{eqn:optimizationP:auxiliary} that binds the variance constraint. Finally, a similar argument to the proof of Proposition 3.1 in \citet{ChiW2020} further shows that the optimal solution to Problem \eqref{eqn:optimizationP:auxiliary} must be unique.

By defining  $U^*(\lambda,\beta):=U_{\lambda,\beta}(I_{\lambda,\beta})$ and for any $\alpha\in[0,1]$, we have
\begin{eqnarray*}
	U^*\big(\alpha\lambda_1+(1-\alpha)\lambda_2,\alpha\beta_1+(1-\alpha)\beta_2\big) & = & \max\limits_{ {I\in \mathcal{IC}}}  U_{\alpha\lambda_1+(1-\alpha)\lambda_2,\alpha\beta_1+(1-\alpha)\beta_2}(I)\\
	&= & \max\limits_{ {I\in \mathcal{IC}}} \big\{\alpha U_{\lambda_1,\beta_1}(I)+(1-\alpha)U_{\lambda_2,\beta_2}(I)\big\} \\
	&\leqslant & \max\limits_{ {I\in \mathcal{IC}}} \big\{\alpha U_{\lambda_1,\beta_1}(I)\}+\max\limits_{ {I\in \mathcal{IC}}} \{(1-\alpha) U_{\lambda_2,\beta_2}(I)\big\} \\
	&= & \alpha \max\limits_{ {I\in \mathcal{IC}}} \big\{U_{\lambda_1,\beta_1}(I)\big\}+(1-\alpha)\max\limits_{ {I\in \mathcal{IC}}} \big\{U_{\lambda_2,\beta_2}(I)\big\} \\
	&= & \alpha U^*(\lambda_1,\beta_1)+(1-\alpha)U^*(\lambda_2,\beta_2),
\end{eqnarray*}
where the second equality is due to the fact that $U_{\lambda,\beta}(I)$ is linear in $(\lambda,\beta)$ for any given $I\in \mathcal{IC}$. Thus, $U^*(\lambda,\beta)$ is convex in $(\lambda,\beta)$.

Furthermore, denoting by $U^{**}$  the maximal EU value of the insured's final wealth in Problem \eqref{eqn:optimizationP:auxiliary}, we have $U^{**}\leqslant U(w_0-(1+\rho)m)< \infty$. On the other hand, for any given $I\in \mathcal{IC}$ satisfying $\EE[I(X)]\neq m$ or $var[I(X)]>\nu$, it is easy to show that $\min\limits_{\lambda\in\RR, \beta\geqslant 0} U_{\lambda,\beta}(I)=-\infty$. Noting that $$ U_{0,0}(I)=\EE[U(w_0-X+I(X)-(1+\rho)m)],$$ we have $$\max\limits_{ {I\in \mathcal{IC}}} \min\limits_{\lambda\in\RR,\beta\geqslant 0} U_{\lambda,\beta}(I)\leqslant U^{**}.$$ Now,
$$
\max_{I\in \mathcal{IC}}U_{\lambda,\beta}(I)\geqslant \max_{I\in \mathcal{IC}_m}U_{\lambda,\beta}(I)\geqslant  \max_{I\in \mathcal{IC}_m}\EE[U(w_0-X+I(X)-(1+\rho)m)],\;\forall \lambda\in\RR,\; \beta\geqslant 0,
$$
leading to
$U^{**}\leqslant \min\limits_{\lambda\in\RR,\beta\geqslant 0} \max\limits_{ {I\in \mathcal{IC}}} U_{\lambda,\beta}(I)$.  Since $U_{\lambda,\beta}(I)$ is continuous and strictly concave in $I$, we can obtain from Lemmas \ref{compact} and \ref{sion} that $$\min\limits_{ {I\in \mathcal{IC}}} \max\limits_{\lambda\in\RR,\beta\geqslant 0} -U_{\lambda,\beta}(I)=\max\limits_{\lambda\in\RR,\beta\geqslant 0} \min\limits_{ {I\in \mathcal{IC}}} -U_{\lambda,\beta}(I),$$ which implies $$\max\limits_{ {I\in \mathcal{IC}}} \min\limits_{\lambda\in\RR,\beta\geqslant 0} U_{\lambda,\beta}(I)=\min\limits_{\lambda\in\RR,\beta\geqslant 0} \max\limits_{ {I\in \mathcal{IC}}} U_{\lambda,\beta}(I)=U^{**}.$$


By denoting $\lambda_U:=(U^{**}+1)/m$, we have $$U^*(\lambda,\beta)= \max\limits_{ {I\in \mathcal{IC}}} U_{\lambda,\beta}(I)\geqslant U_{\lambda,\beta}(0)\geqslant \lambda_U m=U^{**}+1, \;\;\forall \lambda\geqslant \lambda_U,\;\beta\geqslant 0.$$  Furthermore, we define $\lambda_L:=-(U^{**}+1)/\left(\EE[X\wedge K_1]-m\right)$, where $K_1$ is determined by $\EE\left[(X\wedge K_1)^2\right]=\nu+m^2<\EE\left[(X\wedge K_U)^2\right]$. Clearly,  $K_1\in(0,K_U)$. If $\EE[X\wedge K_1]\leqslant m$, then $var[X\wedge K_1]\geqslant \nu=var[X\wedge K_U]$, which contradicts  the definition of $K_U$. Thus, we must have   $\EE[X\wedge K_1]> m$, and hence
$$U^*(\lambda,\beta)= \max\limits_{ {I\in \mathcal{IC}}} U_{\lambda,\beta}(I)\geqslant U_{\lambda,\beta}(x\wedge K_1)\geqslant -\lambda_L\left(\EE[X\wedge K_1]-m\right)\geqslant U^{**}+1,\;\;\forall \lambda\leqslant \lambda_L,\;\beta\geqslant 0.$$ Similarly, let $\beta_U:=\frac{U^{**}+1}{\nu-var[X\wedge K_2]}$, where $K_2$ is determined by $\EE[X\wedge K_2]=m$. Here, we have $K_2< K_U$, which in turn implies $var[X\wedge K_2]<\nu$ based on the definition of $K_U$. For any $\beta\geqslant \beta_U$ and $\lambda\in\RR$, we obtain $$U^*(\lambda,\beta)= \max\limits_{ {I\in \mathcal{IC}}} U_{\lambda,\beta}(I)\geqslant U_{\lambda,\beta}(x\wedge K_2)\geqslant  \beta_U(\nu-var[X\wedge K_2])=U^{**}+1.$$

The above analysis indicates that $$U^{**}= \min_{\lambda_L\leqslant \lambda\leqslant \lambda_U\atop 0\leqslant \beta\leqslant \beta_U} U^*(\lambda,\beta).$$
Thus, it follows from the convexity of $U^*(\lambda,\beta)$ and Weierstrass's theorem that there exist $\lambda^*_m\in [\lambda_L,\lambda_U]$ and $\beta^*_m\in [0,\beta_U]$ such that $$U^{**}=U^*(\lambda^*_m,\beta^*_m)=\min_{\lambda\in\RR, \beta\geqslant 0}U^*(\lambda,\beta).$$ Moreover, we have
\begin{align}\label{dualequation}
U^*(\lambda^*_m,\beta^*_m)=\max_{I\in \mathcal{IC}}U_{\lambda^*_m,\beta^*_m}(I) \geqslant\max_{I\in \mathcal{IC}_m, var[I(X)]=\nu}U_{\lambda^*_m,\beta^*_m}(I)=U^{**},
\end{align}
where the second equality is derived from the fact that the optimal solution to Problem \eqref{eqn:optimizationP:auxiliary} binds the variance constraint.  In addition, thanks to \eqref{dualequation}, the unique optimal solution $I^*_m$ of Problem \eqref{eqn:optimizationP:auxiliary} must solve Problem \eqref{prob:lagr:2} with $\lambda=\lambda^*_m$ and $\beta=\beta^*_m$.  Note that $U_{\lambda^*_m,\beta^*_m}(I)$ is strictly concave in $I$; therefore, $I^*_m(X)=I_{\lambda^*_m,\beta^*_m}(X)$ almost surely. As a result, $I_{\lambda^*_m,\beta^*_m}$  satisfies \eqref{eqn:equality-M-V} and  must be  a solution to Problem \eqref{eqn:optimizationP:auxiliary}.

Finally, we show that $\beta^*_m>0$. Otherwise, if $\beta^*_m=0$, then $I_{\lambda^*_m,\beta^*_m}(x)=(x-d_m)_+$, yielding a contradiction
$$
\nu=var[(X-d_m)_+]> var[(X-d_L)_+]=\nu.
$$
The proof is thus complete. \\

\noindent
\textbf{Proof of Proposition~\ref{lemma:derivative}:} \\

Denote
$$
L_{\beta}(I):=\EE\big[U\big(w_0-X+I(X)-(1+\rho)\EE[I(X)]\big)\big]
-\beta\big(var[I(X)]-\nu\big),\;\beta\geqslant 0,\;I\in \mathcal{IC},
$$
which is linear in $\beta$ and  concave in $I$, because  $U''(\cdot)<0$ and $var[I(X)]$ is convex in $I$.

We denote by $U^{***}$ the maximum EU value of the insured's final wealth in Problem \eqref{eqn:optimal-IC}. Using an argument similar to that in the proof of Proposition~\ref{Pro:optimalS}, we have
$$
\max_{I\in \mathcal{IC}}\min_{\beta\geqslant 0}L_{\beta}(I)\leqslant U^{***}\leqslant \min_{\beta\geqslant 0}\max_{I\in \mathcal{IC}}L_{\beta}(I),
$$
which, together with Lemma \ref{sion}, implies
$$
U^{***}=\min_{\beta\geqslant 0}\max_{I\in \mathcal{IC}}L_{\beta}(I).
$$
Denoting $\tilde{\beta}=\frac{U^{***}+1}{\nu}$, we have
$$
\max_{I\in \mathcal{IC}}L_{\beta}(I)\geqslant L_{\beta}(0)\geqslant \beta\nu\geqslant U^{***}+1,\;\;\forall \beta\geqslant \tilde{\beta}.
$$
Furthermore, since $\max_{I\in \mathcal{IC}}L_{\beta}(I)$ is convex in $\beta$, there must exist  $\beta^*\in[0, \tilde{\beta}]$ such that
$$U^{***}=\max_{I\in \mathcal{IC}}L_{\beta^*}(I).$$ If $\beta^*=0$, then $$U^{***}=\max_{I\in \mathcal{IC}}\EE[U(W_I(X))],$$ in which case the stop-loss insurance $(x-d^*)_+$ is optimal, where $d^*$ is defined in \eqref{eqn:d-star}. However, this is contradicted by Assumption \ref{assump:smallnu}. So we must have $\beta^*>0$.

According to Lemma~\ref{Lemma:m-L}, Lemma~\ref{lemma:m_U} and Proposition~\ref{Pro:optimalS},  $I^*$, which solves Problem \eqref{eqn:simplifiedP} under Assumption \ref{assump:smallnu}, satisfies $var[I^*(X)]=\nu$ and therefore also solves the maximization problem $$\max_{I\in \mathcal{IC}}L_{\beta^*}(I).$$ Thus,  for any $I\in \mathcal{IC}$, we have $$\lim_{\alpha \downarrow 0}\frac{L_{\beta^*}\left((1-\alpha) I^*(x)+\alpha I(x)\right)-L_{\beta^*}(I^*(x))}{\alpha} \leqslant 0,$$
leading to
\begin{align*}
0\geqslant & \EE\Big[U'(W_{I^*}(X))\big(I(X)-I^*(X)-(1+\rho)\EE[I(X)-I^*(X)]\big)\Big]-2\beta^*\big(cov(I^*(X), I(X))-var[I^*(X)]\big)\\
=&\int_0^{\infty}\Big(\EE\big[U'(W_{I^*}(X))(\mathbb{I}_{\{X>t\}}-(1+\rho)\PP\{X>t\})\big] \\
& \ \ \ \ \ \ \   -2\beta^*\big(\EE[I^*(X)\mathbb{I}_{\{X>t\}}]-\EE[I^*(X)]\PP\{X>t\}\big)\Big) \big(I'(t)-{I^*}'(t)\big)\dd t\\
=& \int_0^{\mathcal{M}}\Big(\EE\big[U'(W_{I^*}(X))-2\beta^*I^*(X)|X>t\big]\\
& \ \ \ \ \ \ \ -(1+\rho)\EE\big[U'(W_{I^*}(X))]+2\beta^*\EE[I^*(X)\big]\Big)\PP\{X>t\}\big(I'(t)-{I^*}'(t)\big)\dd t\\
=& \int_0^{\mathcal{M}} \Phi_{I^*}(t) \PP\{X>t\}\big(I'(t)-{I^*}'(t)\big)\dd t,
\end{align*}
where $\Phi_{I^*}$ is defined in \eqref{Phifun} and $\mathbb{I}_A$ is the indicator function of an event $A$. The arbitrariness of  $I\in \mathcal{IC}$ and the fact that $\PP\{X>t\}>0$ for any $t<\mathcal{M}$ yield  that ${I^*}'$ should be in the form of \eqref{NSC}. The proof is complete. \\

\noindent
\textbf{Proof of Theorem \ref{main:theorem}:} \\

Let $I^*$ be optimal for Problem \eqref{eqn:simplifiedP}.
Then it follows from Proposition~\ref{lemma:derivative} that
\begin{align}\label{psi_Phi}
\Phi_{I^*}(x)=\EE\big[\psi(X)|X>x\big]
\end{align}
where
\begin{equation}\label{psi}
\psi(x):=U'(W_{I^*}(x))-2\beta^*I^*(x)-\Big((1+\rho)\EE[U'(W_{I^*}(X))]-2\beta^*\EE[I^*(X)]\Big),
\end{equation}
for some $\beta^*>0$.

Since $F_X$ is assumed to be strictly increasing on $(0,\mathcal{M})$,  it follows from Proposition~\ref{prop:two-point} that $m_L<m_U$.  Recall that solving  Problem \eqref{eqn:simplifiedP} can be reduced to solving Problem \eqref{eqn:three-case-op} under Assumption \ref{assump:smallnu}.
In the following, we proceed to solve Problem \eqref{eqn:three-case-op} with the help of Proposition~\ref{lemma:derivative}. We carry out the analysis for  three cases:
\begin{enumerate}[label=Case (\Alph*)]
	
	\item \label{optimalcase:A} If $I^*(x)=(x-d_L)_+$, then $\psi$ is strictly increasing on $[0, d_L]$ and strictly decreasing on $[d_L,\mathcal{M})$. If there exists  $\tilde{x}\in [d_L,\mathcal{M})$ such that $\psi(\tilde{x})<0$, then $\psi(x)<0$ and thus $\Phi_{I^*}(x)<0$ $\forall x\in [\tilde{x},\mathcal{M}),$ which contradicts Proposition~\ref{lemma:derivative}. Therefore,   $\psi(x)\geqslant 0$ $\forall x\in [d_L,\mathcal{M})$ and  $\psi(d_L)>0$. Because  $\psi(x)$ is continuous in $x$ and $F_X$ is strictly increasing on $(0,\mathcal{M})$, there must exist $\epsilon>0$ such that $\Phi_{I^*}(x)>0$ $\forall x\in[d_L-\epsilon, d_L)$, contradicting Proposition~\ref{lemma:derivative}. So $(x-d_L)_+$  cannot  be an optimal solution to Problem \eqref{eqn:three-case-op}.
	
	\item \label{optimalcase:B} If $I^*(x)=x\wedge K_U$, then $\psi$ is strictly decreasing on $[0, K_U]$ and strictly increasing on $[K_U,\mathcal{M})$. Using a similar argument as the one for \ref{optimalcase:A}, we deduce $\psi(x)\leqslant 0$ $\forall x\in [K_U,\mathcal{M})$ and  $\psi(K_U)<0$.  Since $\psi$ is a continuous function, there exists  $\epsilon>0$ such that $\Phi_{I^*}(x)<0$  $\forall x\in[K_U-\epsilon, K_U)$, contradicting Proposition~\ref{lemma:derivative}. As a result, $x\wedge K_U$  cannot  be optimal  to Problem \eqref{eqn:three-case-op} either.
	
	\item Hence, the optimal solution must be of the form $I^*(x)=I_{\lambda^*_m, \beta^*_m}(x)$ for some $m\in(m_L,m_U)$. Noting that  $I_{\lambda^*_m, \beta^*_m}'(x)\in(0,1)$ for sufficiently large $x$ due to $\beta^*_m>0$ and \eqref{derivative:f}, we deduce   from Proposition~\ref{lemma:derivative} that $\Phi_{I^*}(x)=0$ for sufficiently large $x$. Thus, together with \eqref{solution:equation}, \eqref{psi} and \eqref{psi_Phi}, Proposition~\ref{lemma:derivative} further implies that
\begin{equation}\label{eqn:beta-lambda}
\beta^*_m=\beta^*\qquad\text{and}\qquad \lambda^*_m=(1+\rho)\EE[U'(W_{I^*}(X))]-2\beta^*\EE[I^*(X)].
\end{equation}
 Next, we consider two subcases that depend on the values of $\lambda^*_m$ and $U'(w_0-(1+\rho)m)$.
	
	\begin{enumerate}[label=(C.\arabic*)]
		\item If \label{optimalcase:C1}$\lambda^*_m<U'(w_0-(1+\rho)m)$, then  $I^*(x)=x$  $\forall 0\leqslant x\leqslant \hat{x}$ and
		\[
		U'(w_0-x+I^*(x)-(1+\rho)m)-2\beta^*I^*(x)-\lambda^*_m\left\{\begin{array}{ll}=0,& x\geqslant \hat{x},\\
		>0,&x<\hat{x},\end{array}\right.
		\]
		where $\hat{x}:=\frac{U'(w_0-(1+\rho)m)-\lambda^*_m}{2\beta^*}>0$. Therefore, it follows that
		\begin{eqnarray*}
			0&<&\EE[U'(w_0-X+I^*(X)-(1+\rho)m)-2\beta^* I^*(X)-\lambda^*_m] \\
			&=& \EE[U'(W_{I^*}(X))]-2\beta^* \EE[I^*(X)] -(1+\rho)\EE[U'(W_{I^*}(X))]+2\beta^*\EE[I^*(X)]\\
			&=& -\rho\EE[U'(W_{I^*}(X))],
		\end{eqnarray*}
		leading to a contradiction.

		\item Consequently, we must have $\lambda^*_m\geqslant U'(w_0-(1+\rho)m)$, in which case $I^*$ is coinsurance above a deductible (i.e., \eqref{eqn:f-2}).   Similarly to Subcase \ref{optimalcase:C1}, we can show that
		\[
		U'(w_0-x+I^*(x)-(1+\rho)m)-2\beta^* I^*(x)-\lambda^*_m\left\{
		\begin{array}{ll}
		=0,&x\geqslant \tilde{d},\\
		<0,&x<\tilde{d},
		\end{array}\right.
		\]
		where $\tilde{d}:=w_0-(1+\rho)m-(U')^{-1}(\lambda_m^*)\geqslant 0$. This,  together with \eqref{eqn:beta-lambda}, yields
\begin{equation}\label{eqn:lambda}
\lambda^*_m=U'(w_0-\tilde{d}-(1+\rho)m)
\end{equation}
 and
	\begin{equation}\label{eqn:rho-equal}
			 -\rho\EE[U'(W_{I^*}(X))] =\EE\Big[\big(U'(w_0-X+I^*(X)-(1+\rho)m)-2\beta^*I^*(X)-\lambda^*_m\big)\mathbb{I}_{\{X<\tilde{d}\}}\Big].
	\end{equation}
Therefore, $\rho=0$ if and only if $\tilde{d}=0$.
Moreover,  if $\rho=0$, the above analysis  indicates  that $I^*$ solves the equation  \eqref{eqn:zero-rho}.
 Otherwise, if $\rho>0$, then it follows from \eqref{eqn:beta-lambda} and \eqref{eqn:rho-equal} that
$$
\EE[U'(W_{I^*}(X))]=\frac{\EE[U'(w_0-X-(1+\rho)m)\II_{\{X\leqslant \tilde{d}\}}]+2\beta^*mF_X(\tilde{d})}{1-(1+\rho)\PP\{X>\tilde{d}\}},
$$
which in turn implies $\tilde{d}> VaR_{\frac{1}{1+\rho}}(X)$. Plugging the above equation and \eqref{eqn:lambda} into  \eqref{eqn:beta-lambda} yields
$$
2\beta^*=\frac{1}{m \rho}\left(U'(w_0-\tilde{d}-(1+\rho)m)-(1+\rho)\EE[U'(w_0-X\wedge\tilde{d}-(1+\rho)m)]\right).
$$
As a result, the optimal solution $I^*$  must  be given by \eqref{eqn:I-star}. The proof is complete.
\end{enumerate}
\end{enumerate}

\bigskip

\noindent
\textbf{Proof of Corollary \ref{Vajda}:} \\

We prove for the case in which $\rho>0$, but note that the proof to the case of $\rho=0$ is similar and indeed simpler. It follows from \eqref{eqn:f-2} and \eqref{derivative:f} that $I_{\lambda^*_m, \beta^*_m}(x)=0$ for $x\leqslant \tilde{d}$ and $I'_{\lambda^*_m, \beta^*_m}(x)=f'_{\lambda^*_m, \beta^*_m}(x)$ increases in $x$ for $x>\tilde{d}$, where we use the fact that $x-I_{\lambda^*_m, \beta^*_m}(x)$ increases in $x$, and $\beta^*_m>0$. Moreover, for $x>\tilde{d}$, taking the derivative of $\frac{I_{\lambda^*_m, \beta^*_m}(x)}{x}$ with respect to $x$ yields $$\Big(\frac{I_{\lambda^*_m, \beta^*_m}(x)}{x}\Big)'=\Big(\frac{f_{\lambda^*_m, \beta^*_m}(x)}{x}\Big)'=\frac{f'_{\lambda^*_m, \beta^*_m}(x)x-f_{\lambda^*_m, \beta^*_m}(x)}{x^2}\geqslant \frac{\int_{\tilde{d}}^x \big(f'_{\lambda^*_m, \beta^*_m}(x)-f'_{\lambda^*_m, \beta^*_m}(y)\big)\dd y}{x^2}\geqslant 0.$$
This completes the proof. \\

\noindent
\textbf{Proof of Theorem \ref{CS:initialwealth}:} \\

We first prove the result {\it assuming} that there exists  $x_0\in\left(0,\mathcal{M}\right]$ such that
\begin{equation}\label{eqn:upcrosscondition}
y_0:= e_{I_1^*}(x_0)=e_{I_2^*}(x_0) \quad\text{and}\quad \phi(x_0-y_0)>\frac{\beta_1^*}{\beta_2^*},
\end{equation}
where $\phi$ is defined in \eqref{phi}. First, we show that $e_{I_1^*}(x)$ up-crosses $e_{I_2^*}(x)$ in a neighbour of $x_0$. To this end, we first note
$$
 e_{I_1^*}'(x_0)- e_{I_2^*}'(x_0)=\frac{1}{1+\frac{2\beta_1^*}{-U''(w_1-x_0+y_0)}}-\frac{1}{1+\frac{2\beta_2^*}{-U''(w_2-x_0+y_0)}}>0,
$$
which in turn implies that there exists an $\epsilon>0$ such that
\[
\left\{\begin{array}{ll} e_{I_1^*}(x)< e_{I_2^*}(x),\quad&x\in(x_0-\epsilon, x_0),\\
e_{I_1^*}(x)>e_{I_2^*}(x),\quad&x\in(x_0,x_0+\epsilon).\end{array}\right.\]
Next, we  show that there exists no $y>x_0$ such that $e_{I_1^*}(y)=e_{I_2^*}(y)$. Otherwise, if  such  $y$ existed, then the increasing property of $\phi(z)$ in $z$, together with the strictly increasing property of $x-e_{I_i^*}(x)$ in $x$, would imply $$\phi(y-e_{I_1^*}(y))=\phi(y-e_{I_2^*}(y))>\phi(x_0-e_{I_2^*}(x_0))>\frac{\beta_1^*}{\beta_2^*},$$ which would in turn yield that $e_{I_1^*}$ up-crosses $e_{I_2^*}$ in a neighbour of $y$. This, however, contradicts the fact that $e_{I_1^*}(x)>e_{I_2^*}(x)$ for $x\in(x_0,x_0+\epsilon)$.

Now define
$$
x_1:=\sup\left\{x\in\left[0, x_0\right): e_{I_2^*}(x)<  e_{I_1^*}(x)\right\}.
$$
If $x_1=0$, then $e_{I_1^*}$ up-crosses $e_{I_2^*}$, which contradicts Lemma \ref{distribution:same}. Thus, we must have $0<x_1<x_0$ because $e_{I_1^*}(x)< e_{I_2^*}(x)$ for all $x\in (x_0-\epsilon, x_0)$. Moreover, it follows readily that
$$
\phi\left(x_1-e_{I_1^*}(x_1)\right)\leqslant \frac{\beta_1^*}{\beta_2^*}.
$$
In the following, we show that there exists no point $x_2\in\left(0, x_1\right)$ such that $e_{I_1^*}(x_2)=  e_{I_2^*}(x_2)$.
Indeed, if such $x_2$  existed, then, noting that $\phi$ is strictly increasing, we would have $$\phi\left(x_2-e_{I_1^*}(x_2)\right)<\frac{\beta_1^*}{\beta_2^*},$$ leading to $ e_{I_1^*}'(x_2)-e_{I_2^*}'(x_2)<0.$ In other words, $e_{I_2^*}$ up-crosses $e_{I_1^*}$ at  $x_2$. This contradicts the fact that $e_{I_2^*}$ up-crosses $e_{I_1^*}$ at $x_1$.  Therefore, we can now conclude that $e_{I_2^*}$ up-crosses $e_{I_1^*}$ twice when \eqref{eqn:upcrosscondition} is satisfied.

Let us now consider the case in which \eqref{eqn:upcrosscondition} is not satisfied.  We study two cases:
\begin{enumerate}[label=Case (\Alph*) ]
	\item \label{upcross:casea} If there exists no $x\in \left(0,\mathcal{M}\right]$ such that $e_{I_1^*}(x)=e_{I_2^*}(x)$, then it is easy to show from  $\EE[e_{I_i^*}(X)]=0$ that $e_{I_1^*}(X)$ and $e_{I_2^*}(X)$ have the same distribution. This contradicts Lemma \ref{distribution:same}.
	\item Otherwise, any $x\in \left(0,\mathcal{M}\right]$ satisfying $e_{I_1^*}(x)=e_{I_2^*}(x)$ must have
	$$\phi(x-e_{I_i^*}(x))\leqslant \frac{\beta_1^*}{\beta_2^*}.
	$$
If the above inequality is {\it always} strict, then the previous analysis  shows that $e_{I_2^*}$ up-crosses $e_{I_1^*}$, which contradicts Lemma \ref{distribution:same}. Hence,  there must exist  $x_3\in\left(0,\mathcal{M}\right]$ such that
	$$y_3:=e_{I_1^*}(x_3)=e_{I_2^*}(x_3)\quad\text{and}\quad\phi(x_3-y_3)=\frac{\beta_1^*}{\beta_2^*}.
	$$
	In this case, we further divide our analysis into three subcases.
	\begin{enumerate}[label=(B.\arabic*)]
	\item \label{upcross:caseb1} If $e_{I_2^*}$ up-crosses $e_{I_1^*}$ at  $x_3$, then the above analysis would imply that no up-crossing occurs before $x_3$. A similar argument indicates that $e_{I_1^*}(X)$ and $e_{I_2^*}(X)$ would have the same distribution, which would not be possible.
	\item Otherwise, if $e_{I_1^*}$ up-crosses $e_{I_2^*}$ at  $x_3$, then  the previous analysis shows that $e_{I_2^*}$ up-crosses $e_{I_1^*}$ twice.
	\item Finally, if no up-crossing happens at $x_3$, then we can simply neglect this single point $x_3$ in the analysis. If there further exists $x_4\in (0,x_3)$ satisfying $e_{I_1^*}(x_4)=e_{I_2^*}(x_4)$, then we  have
		$$\phi(x_4-e_{I_i^*}(x_4))< \frac{\beta_1^*}{\beta_2^*}.
		$$
	In this case, the previous analysis indicates that $e_{I_2^*}$ up-crosses $e_{I_1^*}$ at $x_4$, which contradicts Lemma \ref{distribution:same}. Otherwise, if there exists no $x\in \left(0,x_3\right)$ such that $e_{I_1^*}(x)=e_{I_2^*}(x)$, then it follows from  $\EE[e_{I_i^*}(X)]=0$ that $e_{I_1^*}(X)$ and $e_{I_2^*}(X)$ have the same distribution. This again contradicts Lemma \ref{distribution:same}.
	\end{enumerate}
\end{enumerate}

In summary, we have shown that  $e_{I_2^*}$ up-crosses $e_{I_1^*}$ twice. Because $e_{I_1^*}(X)$ and $e_{I_2^*}(X)$ have the same first two moments, we can easily see that the insurer's profit with the wealthier insured,  $-e_{I_2^*}(X)$,  has less downside risk than the counterpart with the less wealthy insured,  $-e_{I_1^*}(X)$,  when the insurance pricing is actuarially fair. The
proof is  complete. \\

\noindent
\textbf{Proof of Corollary \ref{initialwealth:betaandm}:} \\

(i) It follows from the proof of Theorem \ref{CS:initialwealth} that 
$e_{I_2^*}(0)<e_{I_1^*}(0),$ which is equivalent to $\EE[I^*_1(X)]<\EE[I^*_2(X)]$. Suppose that $e_{I_2^*}$ up-crosses $e_{I_1^*}$ at points $x_0$ and $x_1$ with $0<x_0<x_1$. Then $e_{I_1^*}(x_j)=e_{I_2^*}(x_j)$ for $j=0,1$, which, together with \eqref{tilde:equation}, implies
\begin{eqnarray}\nonumber
&& U'(w_1-x_0+e_{I_1^*}(x_0))-U'(w_2-x_0+e_{I_1^*}(x_0))-2(\beta_{1}^*-\beta_2^*)e_{I_1^*}(x_0) \\ \label{equal:11}
&&= U'(w_1-x_1+e_{I_1^*}(x_1))-U'(w_2-x_1+e_{I_1^*}(x_1))-2(\beta_{1}^*-\beta_2^*)e_{I_1^*}(x_1).
\end{eqnarray}
 Denote  $$L(y):=U'(w_1-y)-U'(w_2-y),\;\;y\in[0, w_1).$$  Then  $$L'(y)=-U''(w_1-y)+U''(w_2-y)>0$$ due to $w_1<w_2$ and $U'''>0$. Recalling that $x-e_{I_i^*}(x)$ and $e_{I_i^*}(x)$ are strictly increasing in $x$, we deduce from \eqref{equal:11} that $\beta_2^*<\beta_1^*.$ \\


(ii) Let us denote $\widetilde{w}_i:=w_i-\EE[I^*_i(X)]$, $i=1,2.$ The following analysis depends on the comparison between $\widetilde{w}_1$ and $\widetilde{w}_2$.

\begin{enumerate}[label=Case (\Alph*) ]
	
	\item  If $\widetilde{w}_1=\widetilde{w}_2$ and there exists  $z\in[0,\mathcal{M}]$ such that $I^*_{1}(z)= I^*_{2}(z)$, then it follows from \eqref{eqn:zero-rho} 	
	that $2(\beta^*_{1}-\beta^*_{2})I^*_{i}(z)=0,$ $i=1,2$.  Recalling that $\beta_2^*<\beta_1^*,$ we have $I^*_{i}(z)=0$, which implies $z=0$. Because $\EE[I^*_1(X)]<\EE[I^*_2(X)]$, we conclude that $I^*_{1}(x)<I^*_{2}(x)$  $\forall x\in(0,\mathcal{M}]$.
	
	\item Otherwise, if $\widetilde{w}_1\neq \widetilde{w}_2$, we can show that either $I^*_{1}(x)<I^*_{2}(x)$  $\forall x> 0$ or $I^*_{1}$ up-crosses $I^*_{2}$. Indeed, if there exists  $z\in(0,\mathcal{M}]$ such that $I^*_{1}(z)= I^*_{2}(z)$, then we have
	\begin{eqnarray*}
	&& U'(\widetilde{w}_1-z+I^*_{1}(z))-U'(\widetilde{w}_1)-(U'(\widetilde{w}_2-z+I^*_{2}(z))-U'(\widetilde{w}_2))\\ \nonumber
	&&= 2 (\beta_1^*-\beta_2^*) I^*_{2}(z)>0.
	\end{eqnarray*}
	Noting that $U'(\widetilde{w}-y)-U'(\widetilde{w})$ is strictly decreasing in $\widetilde{w}$ for any $y>0$ because of $U'''>0$, we  must  have $\widetilde{w}_1<\widetilde{w}_2.$  Furthermore, we have
	$$
		 {I_i^*}'(z)= \frac{1}{1+\frac{2\beta_i^*}{-U''(\widetilde{w}_i-z+I_i^*(z))}} = \frac{1}{1+\frac{U'(\widetilde{w}_i-z+I^*_{i}(z))-U'(\widetilde{w}_i)}{-U''(\widetilde{w}_i-z+I_i^*(z))\times I_i^*(z)}},\;\;i=1,2.
$$
	 Denoting  $$H(w):=\frac{U'(w-y)-U'(w)}{-U''(w-y)},\;\;y\in[0, w),$$  we obtain
	\begin{eqnarray*}
		H'(w)&=& \frac{U'(w-y)-U'(w)}{-U''(w-y)} \Big[\mathcal{P}_U(w-y)+\frac{U''(w-y)-U''(w)}{U'(w-y)-U'(w)}\Big] \\
		&=& H(w)\Big[\mathcal{P}_U(w-y)-\mathcal{P}_U(w-\theta y) \Big]>0,
	\end{eqnarray*}
	where the second equality is due to the mean-value theorem with $\theta\in (0,1)$ and the last inequality is due to the assumption of strict DAP. The strictly increasing property of $H(w)$, together with the fact that $\widetilde{w}_1<\widetilde{w}_2$,  yields  ${I_2^*}'(z)<{I_1^*}'(z),$ which, in turn, implies that $I^*_{1}$ up-crosses $I^*_{2}$ at point $z$.

Otherwise, if $I^*_{1}(x)\neq I^*_{2}(x)$ $\forall x\in(0,\mathcal{M})$, then it follows from $\EE[I^*_1(X)]<\EE[I^*_2(X)]$  that $I^*_{1}(x)<I^*_{2}(x)$  $\forall x> 0$. The proof is  complete.
	\end{enumerate}

\noindent
\textbf{Proof of Theorem \ref{CS:variance}:} \\

If there exists no point $x\in (0,\mathcal{M})$ such that $e_{I_1^*}(x)=e_{I_2^*}(x)$, then $e_{I_1^*}(X)$ and $e_{I_2^*}(X)$ are equal in distribution due to $\EE[e_{I_1^*}(X)]=\EE[e_{I_2^*}(X)]=0$. This contradicts the fact that $var[e_{I_1^*}(X)]=\nu_1<\nu_2=var[e_{I_2^*}(X)]$. Therefore, we must have $e_{I_1^*}(z)=e_{I_2^*}(z)$ for some   $z\in (0,\mathcal{M})$; and thus \eqref{tilde:equation:2} implies  $$\EE[U'(w_0-X+e_{I_2^*}(X))]-\EE[U'(w_0-X+e_{I_1^*}(X))]=2(\beta^*_{1}-\beta^*_{2}) e_{I_1^*}(z).$$ We divide the following proof into two cases by comparing  $\EE[U'(w_0-X+e_{I_1^*}(X))]$ with $\EE[U'(w_0-X+e_{I_2^*}(X))]$.

\begin{enumerate}[label=Case (\Alph*) ]
	
	\item \label{variance:case1} If $\EE[U'(w_0-X+e_{I_1^*}(X))] \neq \EE[U'(w_0-X+e_{I_2^*}(X))]$, then $\beta^*_{1}\neq \beta^*_{2}$. Since $e_{I_1^*}$ is a strictly increasing function,  we have $\mathcal{X}:=\{x\in[0,\mathcal{M}]: e_{I_1^*}(x)=e_{I_2^*}(x)\}=\{z\}$.
Recalling that $var[e_{I_1^*}(X)]<var[e_{I_2^*}(X)]$, we conclude  from Lemma \ref{convex:order} that $e_{I_2^*}$ up-crosses $e_{I_1^*}.$
	\item If $\EE[U'(w_0-X+e_{I_1^*}(X))] = \EE[U'(w_0-X+e_{I_2^*}(X))]$, we have
	\begin{align}\label{case:B:result1}
	U'(w_0-x+e_{I_1^*}(x))-2\beta_{1}^*e_{I_1^*}(x)=U'(w_0-x+e_{I_2^*}(x))-2\beta_{2}^*e_{I_2^*}(x).
	\end{align}
	In this case, we further divide our analysis into two subcases based on the comparison between $\beta^*_{1}$ and $\beta^*_{2}$.
	
	\begin{enumerate}[label=(B.\arabic*)]
		\item If $\beta^*_{1}\neq \beta^*_{2}$,
		then it follows from \eqref{case:B:result1} that $e_{I_1^*}(x)=0$ $\forall x\in \mathcal{X}$. Similar to \ref{variance:case1}, we can show that  $e_{I_2^*}$ up-crosses $e_{I_1^*}.$ 	
		
		\item If $\beta^*_{1}= \beta^*_{2}$, then \eqref{case:B:result1} can be rewritten as
		\begin{align*}
		U'(w_0-x+e_{I_1^*}(x))+2\beta_{1}^*(x-e_{I_1^*}(x))=U'(w_0-x+e_{I_2^*}(x))+ 2\beta_{1}^*(x-e_{I_2^*}(x)).
		\end{align*}
		Note that $U'(w_0-y)+2\beta_1^*y$ is strictly increasing in $y$. Hence, the above equation yields $ x-e_{I_1^*}(x)=x-e_{I_2^*}(x)$ for all $x\in[0,\mathcal{M}]$, which contradicts the fact that $var[e_{I_1^*}(X)]<var[e_{I_2^*}(X)]$.
		
	\end{enumerate}
\end{enumerate}
In summary, we conclude  that $e_{I_2^*}$ up-crosses $e_{I_1^*}$.\\

\noindent
\textbf{Proof of Corollary \ref{Coro1:variance}:} \\

(i) From the proof of Theorem \ref{CS:variance}, we know that
\begin{equation}\label{eqn:comparison-I}
e_{I_2^*}(x)-e_{I_1^*}(x)\left\{
\begin{array}{ll}
<0, &x<z;\\
>0,&x\in(z,\mathcal{M}]
\end{array}\right.
\end{equation}
for some $z\in(0,\mathcal{M})$. Therefore, we have  $$\EE[I^*_1(X)]=-e_{I_1^*}(0)<-e_{I_2^*}(0)=\EE[I^*_2(X)].$$ Furthermore,
 $e_{I_1^*}'(z)-e_{I_2^*}'(z)\leqslant 0.$ Recalling that $$e_{I_1^*}'(x)=\frac{1}{1+\frac{2\beta_1^*}{-U''(w_0-x+e_{I_1^*}(x))}},$$ we have $\beta_2^*\leqslant \beta_1^*.$   As $\beta_1^*=\beta_2^*$  does not hold, as  shown  in the proof of Theorem \ref{CS:variance}, we obtain $\beta_2^*<\beta_1^*$.   \\


(ii) On the one hand, in view of \eqref{eqn:comparison-I}, it follows from $\EE[I^*_1(X)]<\EE[I^*_2(X)]$  that $I_2^*(x)> I_1^*(x)$ $\forall x\geqslant z$. On the other hand, for any $x\in [0,z),$ noting that  $e_{I_2^*}(x)< e_{I_1^*}(x)$, we deduce from  $U'''\geqslant 0$ that $$-U''(w_0-x+e_{I_1^*}(x))\leqslant -U''(w_0-x+e_{I_2^*}(x)).$$ Because $$\beta_2^*<\beta^*_1 \quad\text{and} \quad  e_{I_i^*}'(x)=\frac{1}{1+\frac{2\beta_i^*}{-U''(w_0-x+e_{I_i^*}(x))}},\;\;i=1,2,$$ we have $e_{I_1^*}'(x)<  e_{I_2^*}'(x),$ which is equivalent to ${I_1^*}'(x)<{I_2^*}'(x)$   $\forall x\in [0,z)$. Furthermore, as $I_1^*(0)=I_2^*(0)=0$, it must hold that $I_1^*(x)<I_2^*(x)$ $\forall x\geqslant 0$. The proof is complete. \\

\end{document}